%% file: lipics-v2021-article.tex
\documentclass[a4paper,USenglish,cleveref, autoref, thm-restate,pdfa]{lipics-v2021}

\pdfoutput=1 

\usepackage{tikz}
\usepackage{tikz-3dplot}
\usetikzlibrary{shapes, positioning, decorations.pathreplacing, arrows.meta, angles, quotes}
\usepackage{chemfig}
\usepackage[ruled,vlined]{algorithm2e}
\usepackage{complexity}
\newtheorem{problem}{Problem}
\newtheorem*{theorem*}{Theorem}
\newtheorem*{lemma*}{Lemma}

\title{Computational Generation of Substrate-Specific Molecular Cages} 

\titlerunning{Computational Generation of Substrate-Specific Molecular Cages} 

\author{No\'e Demange}{DAVID lab, UVSQ - University Paris-Saclay, Versailles, France}{noe.demange@uvsq.fr}{https://orcid.org/0009-0007-9438-5552}{}

\author{Yann Strozecki}{DAVID lab, UVSQ - University Paris-Saclay, Versailles, France}{yann.strozecki@uvsq.fr}{https://orcid.org/0000-0002-0891-3766}{}

\author{Sandrine Vial}{DAVID lab, UVSQ - University Paris-Saclay, Versailles, France}{sandrine.vial@uvsq.fr}{https://orcid.org/0009-0004-5545-7857}{}

\authorrunning{N. Demange, Y. Strozecki and S. Vial} 

\Copyright{No\'e Demange, Yann Strozecki and Sandrine Vial} 



\category{} 

\relatedversion{} 

\supplementdetails[linktext={\\https://github.com/NoeDemange/MolecularCagesGeneration}, subcategory={Source code}, swhid={swh:1:dir:4724ee746285e388bfbb126e84babf8a21e1aecc}]{Software}{https://github.com/NoeDemange/MolecularCagesGeneration} 

\acknowledgements{Authors want to thank Olivier David for his helpful comments.}

\nolinenumbers 

\EventEditors{Martin Aum\"{u}ller and Irene Finocchi}
\EventNoEds{2}
\EventLongTitle{24th Symposium on Experimental Algorithms (SEA 2026)}
\EventShortTitle{SEA 2026}
\EventAcronym{SEA}
\EventYear{2026}
\EventDate{June 22--24, 2026}
\EventLocation{Copenhagen, Denmark}
\EventLogo{}
\SeriesVolume{371}
\ArticleNo{15}

\newcommand{\rot}[1]{\rotatebox{75}{\scriptsize\textbf{#1}}}
\newcommand{\T}{\rule{0pt}{2ex}}
\newcommand{\coord}{\operatorname{coord}}
\newcommand{\type}{\operatorname{type}}
\newcommand{\col}{\operatorname{col}}
\newcommand{\cov}{\operatorname{cov}}

\newcommand{\ICT}{\operatorname{ICT}}
\DeclareMathOperator{\partition}{Part}

\begin{document}

\definecolor{carbon}{rgb}{0.4, 0.85, 0.4}      
\definecolor{oxygen}{rgb}{0.9, 0.28, 0.28}      
\definecolor{hydrogen}{rgb}{0.8, 0.8, 0.8}    

\definecolor{cubeedge}{rgb}{0.0, 0.0, 0.0}    
\definecolor{gridline}{rgb}{0.7, 0.7, 0.7}    
\definecolor{cubedash}{rgb}{0.0, 0.4, 0.8}    

\maketitle

\begin{abstract}

In this paper, we propose a method to build molecular cages designed to capture
a specific substrate. We model a cage as a graph of atoms with coordinates in space,
and several constraints on their edges (degree, length and angle). We use a simple method to place binding patterns which are able to interact with certain parts of the substrate. We then propose an algorithm which considers all possible ways of connecting these binding patterns and try to construct the smallest possible molecular paths realizing these connections. We investigate many variants of our method in order to obtain the most efficient algorithm, able to build cages of more than a hundred atoms.
\end{abstract}

\section{Introduction}
\label{sec:Introduction}
\input{section/introduction}

\section{Modeling and Definitions}\label{sec:Modeling}

\input{section/model}

\section{Construction of a Molecular Path}\label{sec:path}
\input{section/path}

\section{Interconnection Trees}\label{sec:interconnection}
\input{section/interconnection}

\section{Experimental Evaluation on Real Substrates}\label{sec:evalGlobalOnSub}
\input{section/eval}

\section{Conclusion}\label{sec:conclusion}
\input{section/conclusion}



\bibliography{citations}
\newpage
\appendix

\input{section/appendix}

\end{document}

%% file: section/introduction.tex
Molecular cages are a class of discrete, three-dimensional molecular structures formed by the self-assembly of building blocks into closed frameworks that define an internal cavity capable of hosting guest molecules. These architectures are characterized by well-defined size, shape, and connectivity determined by their constituent components and bonding topology. Since pioneering works of Nobel-prize winning 
chemists D.J. Cram, J.-M. Lehn and C.J. Pedersen~\cite{cram1988, lehn1988,pedersen1988}, recent years have experienced spectacular growth of interest for molecular architectures possessing a defined inner-space~\cite{lehn2001}. The design and synthesis of such molecular cages remain major challenges, as they involve exploring a vast combinatorial space of possible building blocks, connectivities, and three-dimensional arrangements. Anticipating which combinations will lead to a target architecture that can be reliably realized through self-assembly requires navigating a complex design space, where small variations in topology or geometry may result in fundamentally different outcomes. Despite this complexity, successfully synthesized molecular cages exhibit a broad spectrum of applications~\cite{zhang2014}, ranging from the absorption of gases, such as carbon dioxide ($CO_2$) and methane~\cite{holst2010}, to their roles in medicine~\cite{therrien2012,tromans2019} and the containment of hazardous substances, exemplified by the storage of white phosphorus ~\cite{mal2009}.

In this context, cheminformatics and computational modeling provide essential tools to address the complexity of molecular cage design. Although such models offer only coarse 
abstractions of chemical reality, they enable the systematic exploration of large design spaces that are inaccessible to intuition or experimentation alone. In particular, molecular structures are commonly represented as graphs in which vertices represent atoms and edges represent chemical bonds. This graph structure can be further enriched with additional information, such as atom and bond types as well as spatial coordinates. Such graph-based representations, capturing both connectivity and geometry, have long provided a foundation for reasoning about molecular structure, similarity and assembly~\cite{raymond2002,trinajstic2018}.

More recently, advances in algorithmic approaches and data-driven methods, exemplified by breakthroughs such as AlphaFold in protein structure prediction~\cite{jumper2021}, have demonstrated the potential of computational models to provide 
guidance to experimental chemists, even in the presence of significant physical simplifications. Within the community of experimental algorithms, several recent contributions 
have similarly combined heuristic search, combinatorial modeling, and empirical evaluation to study complex chemical problems~\cite{aboulfathSEA2025,gianfrottaSEA2021,antypovSEA2020}.

Two complementary paradigms exist for constructing molecular cages for substrates. A substrate denotes the molecule that is intended to be captured by a molecular cage which is a connected molecular structure that encloses an internal cavity capable of accommodating the substrate. The first paradigm, traditionally used in supramolecular chemistry, follows a host-first strategy: a cage is designed and synthesized a priori, and its ability to encapsulate different substrates is evaluated afterward. This approach has been widely studied, especially for highly symmetric cages that are easier to synthesize and characterize~\cite{barthSEA2015,zhang2014,yang2023}. Although it reliably produces stable and experimentally accessible structures, it often results in limited substrate specificity, since the guest is not considered during the design process. Consequently, selecting an appropriate host for a given target substrate remains largely empirical.

The second paradigm adopts a guest-driven strategy. Starting from a target substrate, the objective is to design a molecular cage that is both compatible with and specific to that substrate. In this context, a molecule is considered as a cage for a given substrate if it is shape-complementary and forms favorable non-covalent interactions with it. Compatibility requires chemical interactions between host and guest, while specificity implies preferential binding, mainly driven by geometric complementarity between the cage interior and the shape of the substrate.

\subsection{Contributions}

We introduce a guest-driven method for generating molecular cages that targets a given substrate. Our approach guarantees substrate specificity by positioning binding patterns around the substrate, following a workflow introduced in~\cite{bricage2018}, as explained in Section~\ref{sec:Modeling}.

We propose several variants of an algorithm to construct molecular paths that connect these binding patterns in Section~\ref{sec:path}. These methods aim to minimize path length, which naturally enforces geometric complementarity between the constructed cage and the shape of the substrate, thereby increasing substrate recognition specificity. Moreover, shorter paths facilitate synthesis in the laboratory, by limiting the structural complexity.

To structure this construction process, we introduce an intermediate combinatorial object in Section~\ref{sec:interconnection}, the interconnection tree, which specifies which binding patterns must be connected. We provide an efficient enumeration algorithm, both in terms of theoretical complexity and practical performance, to generate such trees and to use them as a backbone for cage assembly.

Beyond the proposed modeling framework and algorithms, a major contribution of this work lies in the extensive experimental evaluation of our methods. In Section~\ref{sec:evalGlobalOnSub}, we systematically explore a wide range of parameters, including discretization levels, distance-based construction heuristics, branching factor in exhaustive search, and the use of interconnection tree weights for guidance. These experiments are conducted on real molecular substrates extracted from the Cambridge Structural Database~\cite{CSD}, allowing us to identify parameter settings that are effective in practice. As a result, our approach can automatically generate molecular cages for a wide diversity of substrates within reasonable computation times.

%% file: section/model.tex
This section introduces the core objects used throughout the paper. Since we operate on molecular structures, we adopt a graph-based model that explicitly incorporates chemical constraints. Atoms and covalent bonds are represented as vertices and edges, respectively, and the embedding of the graph in three-dimensional space is used to capture geometric and chemical feasibility.

In general, substrates may contain a wide variety of atom types. In this work, we restrict ourselves to organic molecules. For the construction of molecular cages, we consider a limited set of atom types, namely carbon, hydrogen, oxygen, and nitrogen, with a predominance of carbon and hydrogen. This restriction simplifies the model while remaining expressive enough for the targeted applications. The algorithms themselves are generic and can be extended to additional atom types.

We introduce the concept of a \textbf{molecular graph}, which encodes all the relevant information about the molecules under study (see an example in Figure~\ref{fig:MGall_BAHSUY}).

\newcommand{\Ca}[3]{\node[ball color=carbon, circle, minimum size=10pt] at (#1, #2, #3) {\tiny C};}
\newcommand{\Ox}[3]{\node[ball color=oxygen, circle, minimum size=10pt] at (#1, #2, #3) {\tiny O};}
\newcommand{\Hy}[3]{\node[ball color=hydrogen, circle, minimum size=10pt] at (#1, #2, #3) {\tiny H};}

\begin{figure}[b]
    \centering
    \begin{subfigure}[b]{0.59\textwidth}
        \centering
        \resizebox{0.75\linewidth}{!}{
            \begin{tikzpicture}
                \node[circle, fill=carbon, minimum size=1cm, text=black] (C1) at (0,0){C};
                \node[circle, fill=hydrogen, minimum size=1cm, text=black] (H6) at (-2.2,0){H}; 
                \node[circle, fill=hydrogen, minimum size=1cm, text=black] (H7) at (0,2.2){H}; 
                \node[circle, fill=hydrogen, minimum size=1cm, text=black] (H8) at (0,-2.2){H}; 
                
                    \draw (C1) -- (H6)node[midway, sloped, above, font=\footnotesize] {$0.11\,\mathrm{nm}$};
                    \draw (C1) -- (H7)node[midway, sloped, above, font=\footnotesize] {$0.11\,\mathrm{nm}$};
                    \draw (C1) -- (H8)node[midway, sloped, above, font=\footnotesize] {$0.11\,\mathrm{nm}$};

                \node[circle, fill=oxygen, minimum size=1cm, text=black] (O1) at (3,0){O}; 
                \draw (C1) -- (O1)node[midway, sloped, above, font=\footnotesize] {$0.15\,\mathrm{nm}$};
                \node[circle, fill=carbon, minimum size=1cm, text=black] (C3) at (6,0){C}; 
                \node[circle, fill=carbon, minimum size=1cm, text=black] (C5) at (6,-3){C}; 
                \node[circle, fill=oxygen, minimum size=1cm, text=black] (O4) at (9,0){O}; 
                \draw (C3) -- (O1)node[midway, sloped, above, font=\footnotesize] {$0.15\,\mathrm{nm}$};
                \draw (C3) -- (O4)node[midway, sloped, above, font=\footnotesize] {$0.15\,\mathrm{nm}$};
                \draw (C3) -- (C5)node[midway, sloped, above, font=\footnotesize] {$0.15\,\mathrm{nm}$};

                \node[circle, fill=hydrogen, minimum size=1cm, text=black] (H9) at (3.8,-3){H}; 
                \node[circle, fill=hydrogen, minimum size=1cm, text=black] (H10) at (6,-5.2){H}; 
                \node[circle, fill=hydrogen, minimum size=1cm, text=black] (H11) at (8.2,-3){H}; 
                    \draw (C5) -- (H9)node[midway, sloped, above, font=\footnotesize] {$0.11\,\mathrm{nm}$};
                    \draw (C5) -- (H10)node[midway, sloped, above, font=\footnotesize] {$0.11\,\mathrm{nm}$};
                    \draw (C5) -- (H11)node[midway, sloped, above, font=\footnotesize] {$0.11\,\mathrm{nm}$};

                \draw pic[draw,"$109.5^\circ$", angle radius=0.6cm,angle eccentricity=1.45, pic text options={xshift=-5pt,yshift=-1pt}, font=\footnotesize]{angle = H7--C1--H6};
                \draw pic[draw,"$109.5^\circ$", angle radius=0.75cm,angle eccentricity=1.45, font=\footnotesize]{angle = H6--C1--H8};
                \draw pic[draw,"$109.5^\circ$", angle radius=0.75cm,angle eccentricity=1.55, font=\footnotesize]{angle = O1--C1--H7};
                \draw pic[draw,"$120^\circ$", angle radius=0.75cm,angle eccentricity=1.45, font=\footnotesize]{angle = C1--O1--C3};
                \draw pic[draw,"$120^\circ$", angle radius=0.75cm,angle eccentricity=1.35, font=\footnotesize]{angle = O4--C3--O1};
                \draw pic[draw,"$120^\circ$", angle radius=0.6cm,angle eccentricity=1.45, font=\footnotesize]{angle = O1--C3--C5};
                \draw pic[draw,"$109.5^\circ$", angle radius=0.6cm,angle eccentricity=1.55, font=\footnotesize]{angle = C3--C5--H9};
                \draw pic[draw,"$109.5^\circ$", angle radius=0.75cm,angle eccentricity=1.45, font=\footnotesize]{angle = H9--C5--H10};
                \draw pic[draw,"$109.5^\circ$", angle radius=0.6cm,angle eccentricity=1.45, pic text options={xshift=8pt,yshift=2pt}, font=\footnotesize]{angle = H10--C5--H11};
            \end{tikzpicture}
        }
    \end{subfigure}
    \begin{subfigure}[b]{0.39\textwidth}
        \centering
        \resizebox{\linewidth}{!}{
            \begin{tikzpicture}[scale=1]
                \tdplotsetmaincoords{75}{120}
                \begin{scope}[tdplot_main_coords]
            
                    \coordinate (A) at (0,0,0);
                    \coordinate (B) at (3,0,0);
                    \coordinate (C) at (3,5,0);
                    \coordinate (D) at (0,5,0);
                    \coordinate (E) at (0,0,3);
                    \coordinate (F) at (3,0,3);
                    \coordinate (G) at (3,5,3);
                    \coordinate (H) at (0,5,3);
        
                    \foreach \i in {0,1,...,3} {
                      \draw[gridline,dashed] (A) ++(0,0,\i) -- +(3,0,0);
                      \draw[gridline,dashed] (A) ++(\i,0,0) -- +(0,0,3);
                      \draw[gridline,dashed] (A) ++(\i,0,0) -- +(0,5,0);
                    }
                    
                    \foreach \i in {0,1,...,5} {
                      \draw[gridline,dashed] (A) ++(0,\i,0) -- +(3,0,0);
                    }
        
                    \draw[cubeedge,dashed] (A) -- (B) ;
                    \draw[cubeedge] (B) -- (C);
                    \draw[cubeedge] (C) -- (D);
                    \draw[cubeedge,dashed] (A) -- (D);
                    \draw[cubeedge] (E) -- (F);
                    \draw[cubeedge] (G) -- (H);
                    \draw[cubeedge] (E) -- (H);
                    \draw[cubeedge,dashed] (A) -- (E);
                    \draw[cubeedge] (B) -- (F);
                    \draw[cubeedge] (D) -- (H);
            
                    \draw[->,dashed] (A) -- (B) node[anchor=north east] {$x$};
                    \foreach \i in {0,1,2,3} {
                      \draw (A) ++(\i,0,0) -- ++(0,0,0.15)
                            node[anchor=south east] {\footnotesize \i};
                    }
            
                    \draw[->,dashed] (A) -- (D) node[anchor=south west] {$y$};
                    \foreach \i in {0,1,2,3,4,5} {
                      \draw (A) ++(0,\i,0) -- ++(0.15,0,0)
                            node[anchor=south] {\footnotesize \i};
                    }
            
                    \draw[->,dashed] (A) -- (E) node[anchor=south] {$z$};
                    \foreach \i in {0,1,2,3} {
                      \draw (A) ++(0,0,\i) -- ++(0.15,0,0)
                            node[anchor=north west] {\footnotesize \i};
                    }

                    \draw[thick] (1.4352,0.2875,1.7920)--(1.3999,1.7352,1.9103); 
                    \draw[thick] (1.4352,0.2875,1.7920)--(0.5832,0.0000,1.8960);
                    \draw[thick] (1.4352,0.2875,1.7920)--(1.8879,0.0557,0.9917);
                    \draw[thick] (1.4352,0.2875,1.7920)--(1.9487,0.1462,2.5818);
                    \draw[thick] (1.3999,1.7352,1.9103)--(0.7516,2.3644,0.9239);
                    \draw[thick] (0.7516,2.3644,0.9239)--(0.2582,1.7790,0.0000);
                    \draw[thick] (0.7516,2.3644,0.9239)--(0.7725,3.8461,1.1047);
                    \draw[thick] (0.7725,3.8461,1.1047)--(0.7857,4.0995,1.9563);
                    \draw[thick] (0.7725,3.8461,1.1047)--(0.0000,4.1691,0.6601);
                    \draw[thick] (0.7725,3.8461,1.1047)--(1.5148,4.0925,0.5546);
            
                    \Hy{0.5832}{0.0000}{1.8960}
                    \Hy{1.8879}{0.0557}{0.9917}
                    \Hy{1.9487}{0.1462}{2.5818}
                    \Hy{0.7857}{4.0995}{1.9563}
                    \Hy{0.0000}{4.1691}{0.6601}
                    \Hy{1.5148}{4.0925}{0.5546}
        
                    \Ca{1.4352}{0.2875}{1.7920}
                    \Ca{0.7516}{2.3644}{0.9239}
                    \Ca{0.7725}{3.8461}{1.1047}
        
                    \Ox{1.3999}{1.7352}{1.9103}
                    \Ox{0.2582}{1.7790}{0.0000}
            
                    \draw (C) -- (G);
                    \draw (F) -- (G);
            
                \end{scope}
            \end{tikzpicture}
        }
    \end{subfigure}
    \caption{Two views of the molecular graph representing Methyl Acetate. The left view is the abstract graph, while the right one represents the spatial positions of the vertices.}
    \label{fig:MGall_BAHSUY}
\end{figure}
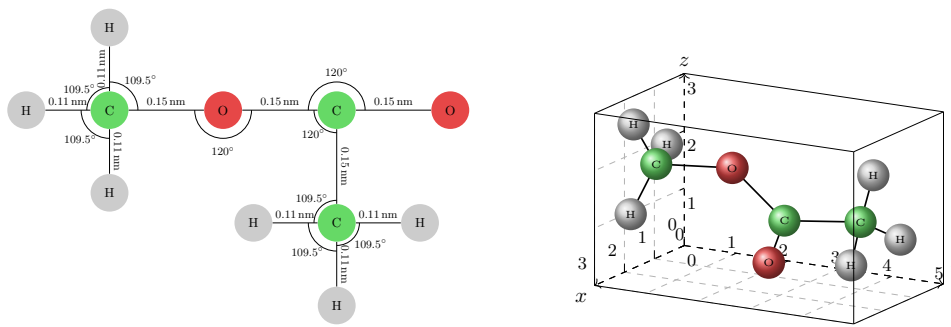

\begin{definition}
A \textbf{molecular graph} is a triple $(G, \coord, \type)$, where $G = (V, E)$ is a graph with vertex set $V$ (atoms) and edge set $E$ (covalent bonds). The function $\coord : V \rightarrow \mathbb{R}^3$ assigns coordinates to each vertex, and the function $\type : V \rightarrow A$ assigns to each vertex an atom type from a finite set $A$.
\end{definition}


A \textit{covalent bond} corresponds to a pair of atoms sharing electrons, and its length depends on the types of the atoms involved. We denote by $\cov(a_1,a_2)$ the length of the covalent bond of atoms of types $a_1$ and $a_2$. In this article, we use a simplified model with two bond lengths: $0.1125\,\mathrm{nm}$ for bonds involving hydrogen atoms, and $0.15\,\mathrm{nm}$ for bonds between non-hydrogen atoms.

The collision distance models steric repulsion, which prevents atoms from being placed too close to each other due to overlap of their electron clouds. We denote by $\col(a_1,a_2)$ the collision distance of atoms of types $a_1$ and $a_2$. In this article, we use a uniform collision distance of $0.1125\,\mathrm{nm}$.

The degree of a vertex corresponds to the number of covalent bonds it can form and is bounded by the valence of its atom type. Since we consider organic molecules, this degree is at most 4, but may be smaller depending on the atom type (see Table~\ref{tab:degree}). Formally, the degree of a vertex $v \in V$ is bounded by $\operatorname{d}(\type(v))$.

To lighten the presentation, we denote the distance between two vertices by $||u - v||$ instead of
$||\coord(u) - \coord(v)||$.

We define the notion of a \textbf{chemically realistic} molecular graph to capture the set of chemical constraints that such a graph must satisfy.

\begin{definition}
A molecular graph is chemically realistic if it satisfies, the following constraints:
\begin{itemize}
    \item No atomic collision: $\forall u, v \in V,\; \|u-v\| > \col(\type(u), \type(v))$.
    \item Covalent bond length: $\forall (u,v) \in E,\; \|u-v\| = \cov(\type(u), \type(v))$. 
    \item Valence constraint: $\forall v \in V,\; \operatorname{deg}(v) \leq \operatorname{d}(\type(v))$. 
    \item Bond angles around each vertex respect the VSEPR theory~\cite{gillespie1963}, with the relevant values given in Table~\ref{tab:VSEPR}.
\end{itemize}
\end{definition}

\begin{table}[t]
\centering
\scriptsize

\begin{minipage}[t]{0.25\textwidth}
\centering
\caption{Number of bonds depending on the type of atoms}
\label{tab:degree}
\setlength{\tabcolsep}{2pt}
\begin{tabular}{|c|c|}
\hline
\textbf{Atom} & \textbf{\#bonds} \\ \hline
Carbon   & 4 \\ \hline
Nitrogen & 3 \\ \hline
Oxygen   & 2 \\ \hline
Hydrogen & 1 \\ \hline
\end{tabular}
\end{minipage}
\hfill
\begin{minipage}[t]{0.71\textwidth}
\centering
\tiny
\caption{Representation and angle constraint based on geometry according to VSEPR theory~\cite{gillespie1963}.}
\label{tab:VSEPR}
\setlength{\tabcolsep}{2pt}
\begin{tabular}{|c|c|c|c|c|}
\hline
\textbf{Geometry} & \textbf{2D} & \textbf{3D} & \textbf{Angle} & \textbf{Margin} \\ \hline
Tetrahedral &
\chemfig[atom sep=0.7cm]{A(-[:90]X)(-[:199.5]X)(<[:309]X)(<:[:340.5]X)} &
\raisebox{-.4\height}{\includegraphics[scale=0.07]{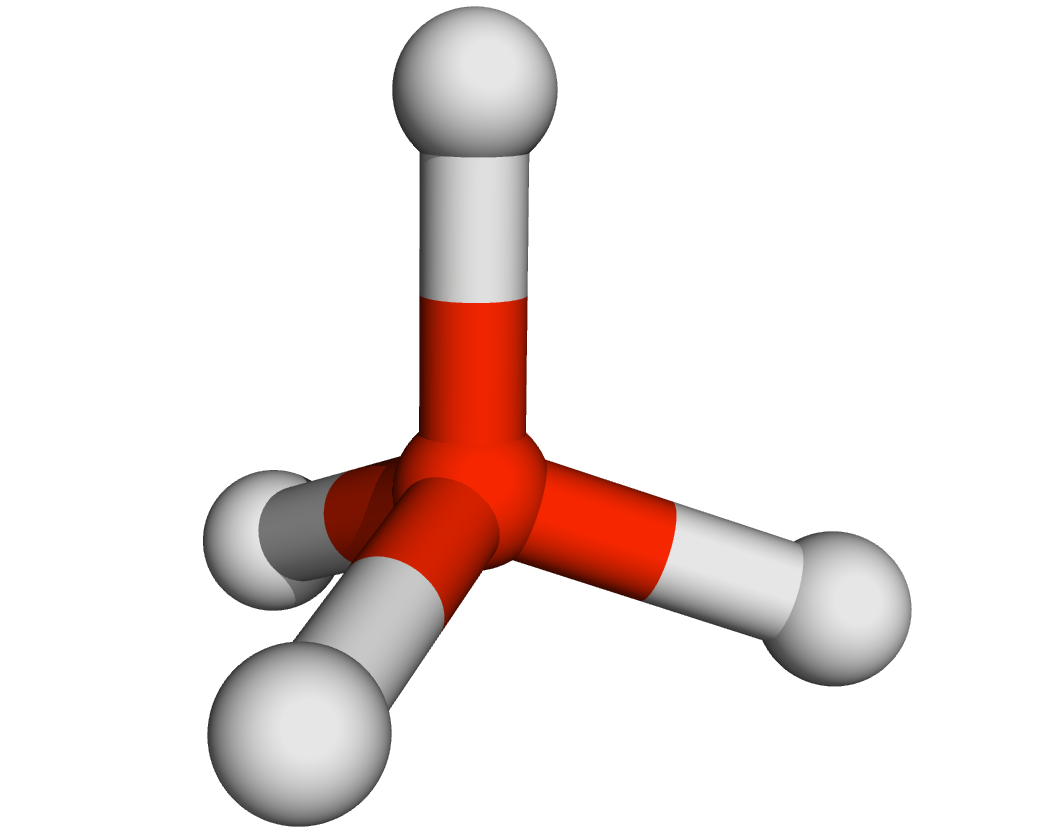}} &
109.5° & 3° \\ \hline
Triangular &
\chemfig[atom sep=0.7cm]{A(-[:180]X)(<[:330]X)(<:[:30]X)} &
\raisebox{-.4\height}{\includegraphics[scale=0.06]{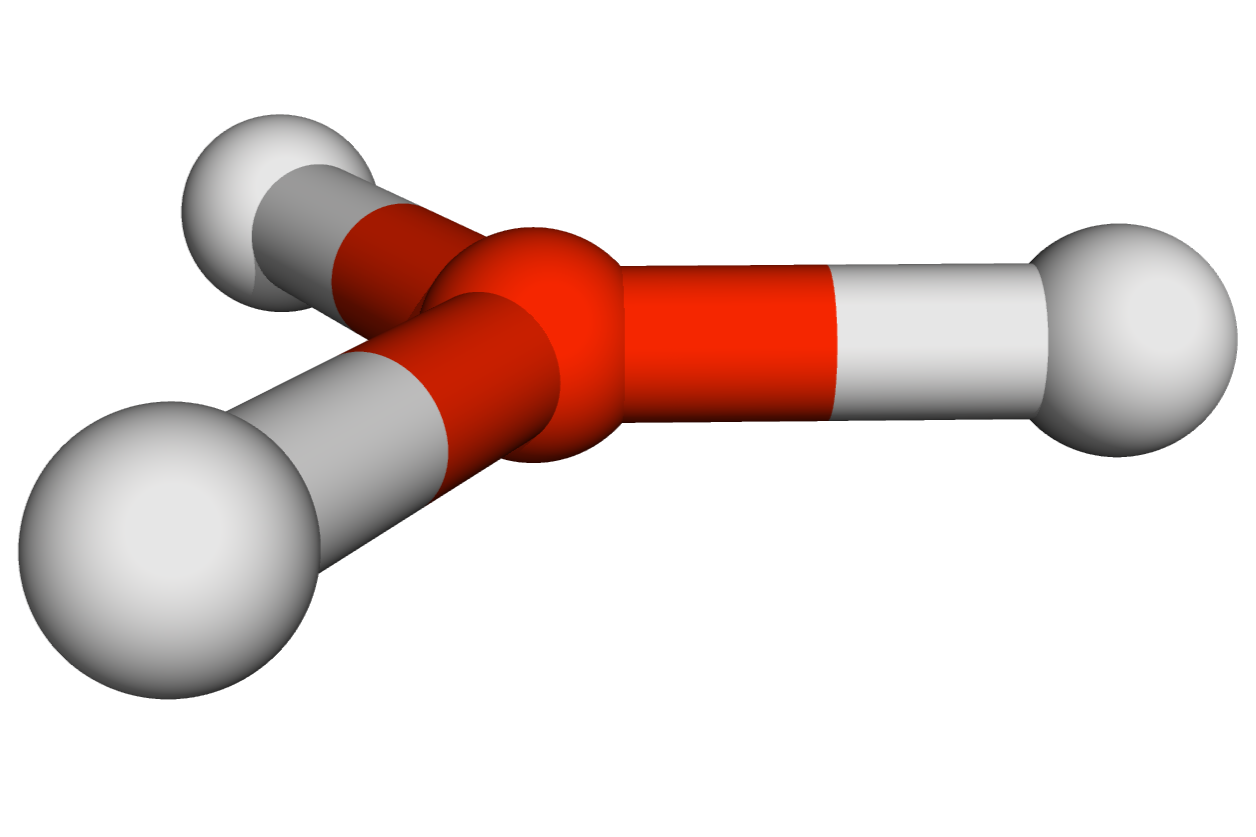}} &
120° & 2° \\ \hline
Linear &
\chemfig[atom sep=0.7cm]{X-A-X} &
\raisebox{-.2\height}{\includegraphics[scale=0.07]{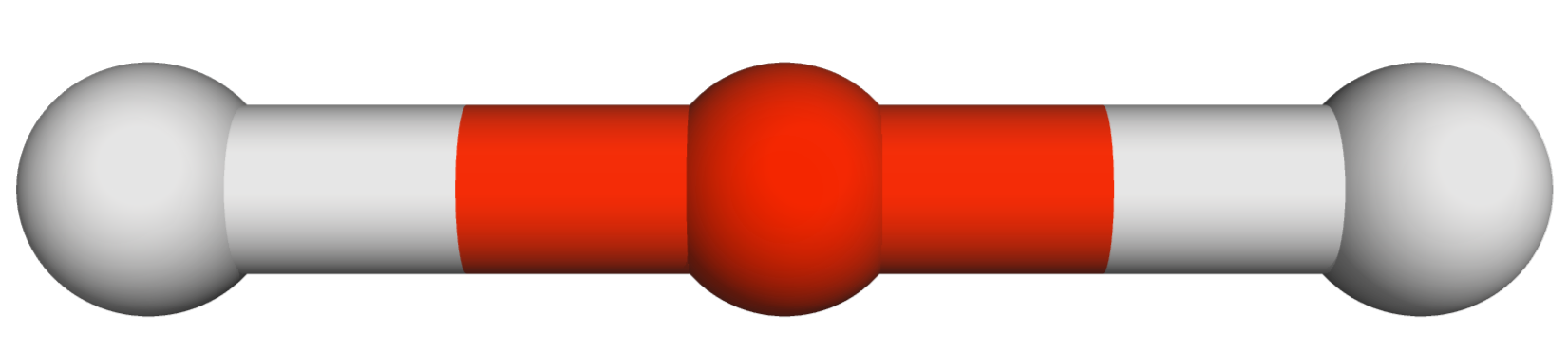}} &
180° & -- \\ \hline
\end{tabular}
\end{minipage}

\end{table}

\paragraph*{Binding Patterns}

Efficient substrate recognition requires a molecular cage to establish multiple intermolecular interactions in order to ensure both binding strength and selectivity. We focus on organic substrates and consider two key classes of non-covalent interactions: hydrogen bonds and $\pi$–$\pi$ stacking interactions.

A \textit{hydrogen bond}~\cite{arunan2011} involves a hydrogen atom bound to an electronegative donor and a second electronegative acceptor atom, typically oxygen or nitrogen. Its geometry is constrained in both distance and orientation (Figure~\ref{fig:hydrogen}). In our model, when the substrate provides a donor (resp.\ acceptor), we introduce a complementary acceptor (resp.\ donor) as a \textit{hydrogen binding pattern}. An atom set can act as donor or acceptor, but not both, and can participate in at most one hydrogen bond.

A \textit{$\pi$–$\pi$ stacking interaction}~\cite{meyer2003} occurs between approximately planar and parallel aromatic systems at a characteristic distance. We model such interactions by planar cyclic subgraphs representing aromatic rings (Figure~\ref{fig:aromatic}).

\begin{figure}[t]
\centering

\begin{minipage}[t]{0.49\textwidth}
\centering
\resizebox{0.6\linewidth}{!}{
\begin{tikzpicture}

\node[circle, fill=oxygen, minimum size=1cm, draw=blue, text=black] (O1) at (0,0) {\Large \textbf{O}};
\node[circle, fill=hydrogen, minimum size=1cm, draw=blue, thick] (H1) at (2.2,0) {\Large \textbf{H}};
\node[circle, fill=hydrogen, minimum size=1cm, draw=blue, thick] (H11) at (-1.41,-1.69) {\Large \textbf{H}};

\draw[blue] (O1) -- (H1);
\draw[blue] (O1) -- (H11);

\node[text=blue] at (-0.3,-1) [right] {\LARGE \textbf{Donor}};

\node[circle, fill=oxygen, minimum size=1cm, text=black] (O2) at (5.8,0) {\Large \textbf{O}};
\node[circle, fill=hydrogen, minimum size=1cm, draw=black, thick] (H2) at (6.9,1.91) {\Large \textbf{H}};
\node[circle, fill=hydrogen, minimum size=1cm, draw=black, thick] (H22) at (6.9,-1.91) {\Large \textbf{H}};

\draw (O2) -- (H2);
\draw (O2) -- (H22);

\node at (6.2,0) [right] {\LARGE \textbf{Acceptor}};

\draw[->,>=stealth,dashed, line width=1mm, color=orange] (H1) -- (O2);
\node at (3.8,1) [color=orange] {\LARGE\textbf{Hydrogen bond}};

\end{tikzpicture}
}
\caption{Hydrogen bond between two water molecules.}
\label{fig:hydrogen}
\end{minipage}
\hfill
\begin{minipage}[t]{0.49\textwidth}
\centering
\resizebox{0.42\linewidth}{!}{
\begin{tikzpicture}

\node (ring1) at (0,0) {
  \chemfig[cram width=2pt]{
    ?<[7,0.7]-[,,,,line width=2pt]
    >[1,0.7]-[3,0.7]-[4]?
  }
};

\draw[dashed, line width=0.8mm, color=orange] (0,0) -- (0,1.8);
\node at (2,0.9) [color=orange] {\LARGE\textbf{$\pi$--$\pi$ stacking}};

\node (ring2) at (0,1.8) {
  \chemfig[bond style={,red},cram width=2pt]{
    ?<[7,0.7]-[,,,,line width=2pt]
    >[1,0.7]-[3,0.7]-[4]?
  }
};

\end{tikzpicture}
}
\caption{$\pi$--$\pi$ stacking interactions between aromatic rings, which are $5$ or $6$-cyles of carbons.}
\label{fig:aromatic}
\end{minipage}

\end{figure}

Binding pattern placement follows the approach of Bricage~\cite{bricage2018}. For a given substrate, multiple valid placements may exist. We first identify all substrate sites eligible for hydrogen bonding or $\pi$–$\pi$ stacking based on atom types and local geometry. For each site, we generate candidate placements satisfying the geometric constraints of the interaction and discard those that collide with the substrate.

$\pi$–$\pi$ stacking patterns are placed first, as they are spatially extended, strongly constrain the geometry, and rarely interfere with one another. Hydrogen patterns are then considered. Their placement is subject to geometric and chemical incompatibilities, which are encoded in a conflict graph whose vertices represent patterns and edges represent conflicts. Selecting a compatible set of patterns corresponds to computing an independent set in this graph. We enumerate all maximal independent sets using the Bron–Kerbosch algorithm~\cite{bron1973}. Maximal independent sets correspond to sets of binding patterns that cannot be extended by adding another compatible pattern, and therefore represent configurations that maximize the number of interactions. In practice, we generate all maximal independent sets and then select representative ones to construct candidate cages. In future work, one could systematically evaluate all maximal sets or also consider non-maximal sets when a smaller number of interactions is desired.

This process yields a disconnected molecular graph of mutually compatible binding patterns, which are in interaction with the substrate. The graph is typically disconnected, since the binding patterns do not form covalent bonds between each other at this stage. The next step is therefore to connect these patterns using molecular paths in order to obtain a connected molecular cage. Section~\ref{sec:path} describes how these patterns can be connected by molecular paths to form a complete cage while satisfying chemical and geometric constraints.

\paragraph*{Connection Patterns and Paths}

Cage construction relies on basic building blocks, called \textbf{connection patterns}. A connection pattern is a molecular graph with two distinguished vertices, denoted $v_{in}$ and $v_{out}$, corresponding to atoms that can form covalent bonds with other patterns. These vertices define entry and exit points, allowing patterns to be assembled into molecular paths that connect binding patterns. The vertices $v_{in}$ and $v_{out}$ may coincide, as in the case of a single central atom, such as the tetrahedral carbon extensively used in this work.

\begin{definition}
Let $(p_1, p_2, \dots, p_k)$ be a sequence of connection patterns, and let $v_{in}^i$ and $v_{out}^i$ denote the distinguished vertices of $p_i$. We define $P$ as the union of the graphs $p_i$, augmented with edges $(v_{out}^i, v_{in}^{i+1})$ for all $i < k$. If $P$ is a molecular graph, the sequence $(p_1, p_2, \dots, p_k)$ is called a \textbf{molecular path}.
\end{definition}

%% file: section/path.tex
In this section, we introduce the concept of constructing a molecular path and the underlying \textsc{Molecular Path Construction} problem.

We consider two molecular graphs: $G_S$ representing a substrate and $G$ representing binding patterns and a partially constructed cage, together with two vertices $s$ and $t$ of $G$ corresponding to atoms of two binding patterns to be connected. A molecular path $P = (p_1,\dots,p_k)$ connects $s$ to $t$ if it can be attached to $G$ at these vertices and satisfies the following structural constraints:
\begin{enumerate}
\item $G \cup P$ is a chemically realistic molecular graph;
\item let $v_{in}^1$ be the distinguished vertex of $p_1$,
$\|v_{in}^1 - s\| = \cov(\type(v_{in}^1), \type(s))$;
\item let $v_{out}^k$ be the distinguished vertex of $p_k$, $\|v_{out}^k - t\| = \cov(\type(v_{out}^k), \type(t))$;
\item for all $u \in V_{G_S}$ and all $v \in V_P$,
$\|u - v\| \geq d_{weak}$.
\end{enumerate} 

The first constraint ensures chemical plausibility of the constructed structure, while constraints (2) and (3) enforce the attachment of the path $P$ to the binding patterns at atoms $s$ and $t$.
A \textit{weak bond} is an attractive, non-covalent interaction either within a molecule or between molecules. While different types of weak interactions exist and may correspond to different distances, we adopt a simplified model in which such interactions occur at a fixed distance of $0.18\,\mathrm{nm}$. Accordingly, in Constraint~(4), $d_{weak}$ denotes the minimum distance required for weak interactions. 
Since $d_{weak}$ is chosen such that $\forall a_1, a_2 \in A, d_{weak}>col(a_1,a_2)$, this constraint is stronger than the non-collision requirement implied by Constraint~(1) and ensures that the only weak interactions present are those explicitly modeled by the binding patterns.

In practice, constructing a molecular path that exactly satisfies all geometric constraints is generally impossible, due to the flexibility of molecular conformations. We therefore allow small deviations at the final attachment to $t$, bounded by $10^{\circ}$ for bond angles and $0.05\,\mathrm{nm}$ for bond lengths. These tolerances are fixed for simplicity and because the considered paths are short (fewer than 15 atoms). For longer paths, it would be natural to scale them with path length.
Rather than distributing these deviations along the entire structure, as in real molecules, we approximate this behavior by concentrating them at a single atom, which simplifies both construction and analysis.

We quantify the resulting distortion using the normalized root mean square deviation (NRMSD) of bond angles relative to VSEPR theory and of bond lengths relative to ideal covalent distances, providing a global measure of deviation from an ideal configuration.
The terminal vertex must satisfy the above angular and bond length tolerances with respect to the target, yielding a worst-case NRMSD of~$\sqrt{3}$, computed as {\small\[\sqrt{\left(\frac{\measuredangle C_{p_{k-2}}C_{p_{k-1}}C_{p_{k}} - 109.5}{10}\right)^2+\left(\frac{\measuredangle C_{p_{k-1}}C_{p_{k}}C_{p_{k+1}} - 109.5}{10}\right)^2+\left(\frac{\lVert C_{p_{k-1}} - C_{p_{k}}\rVert - 0.15}{0.05}\right)^2}\]}

We now define the \textsc{Molecular Path Construction} problem as follows: given the substrate $G_S$, a partial cage $G$, and two vertices $s,t \in V(G)$, find a molecular path $P$ connecting $s$ to $t$ that minimizes $(|P|,\mathrm{NRMSD}(P))$ in lexicographic order, where $|P|$ denotes the length of $P$.

For clarity of presentation, we restrict \textsc{Molecular Path Construction} to a single type of connection pattern: a tetrahedral carbon atom (see Figure~\ref{fig:pos_ajout_with_hydro}). This choice is sufficient to generate realistic cages, and all proposed methods naturally extend to multiple pattern types.

\subsection{Valid Positions of the Next Connection Pattern}

To solve the \textsc{Molecular Path Construction} problem, we propose a depth-first exploration of the solution space, guided by distance-based heuristics. The molecular path is constructed incrementally by successively adding connection patterns: At each step, a new pattern $p$ is appended to the end of a partial path $(p_1,\dots,p_i)$, producing a longer partial solution.  

A key challenge arises from the continuous nature of 3D space: even when geometric constraints are fixed, there is an infinite number of admissible positions for the next atom. To make the search tractable, we first discretize the space of possible positions around the relevant edge rotation. In this sense, the algorithm can be interpreted as a branch-and-bound search over the discretized rotation space, where partial paths are extended only along positions that are geometrically feasible and heuristically promising.

The main problem addressed in the remainder of this section is how to generate well-distributed, legal positions for the next connection pattern and how to select among them those that are most likely to yield shorter, chemically plausible paths.

First, we explain how to determine the legal positions to add an atom, according to VSEPR theory and fixed bond lengths. Once two consecutive atoms are known, the next atom must lie on a circle in 3D space (Figure~\ref{fig:pos_ajout}). 

When the position of this next atom is fixed, the two missing hydrogen atoms required to complete the tetrahedral carbon pattern, centered on $C_{p_i}$, the central atom of the pattern $p_i$, are added automatically (Figure~\ref{fig:pos_ajout_with_hydro}). We must also check whether these hydrogens satisfy collision constraints. Thus, placing a single atom induces the placement of two additional atoms, increasing the complexity of collision detection.

\begin{figure}[b]
\centering
    \begin{subfigure}[t]{0.47\textwidth}
        \centering
        \resizebox{0.65\linewidth}{!}{
        \tdplotsetmaincoords{-10}{0}
        
        \begin{tikzpicture}[tdplot_main_coords, scale=1]
            \definecolor{myBlue}{RGB}{70,200,220}

            \coordinate (A) at (0, 0, 0); 
            \coordinate (top) at (-2.83, 1, 0);
            \coordinate (bot) at (0, -3, 0);
        
            \draw[thick] (A) -- (bot);
            \draw[thick] (A) -- (top);
            
            \draw[very thick, cyan] (0, 1, 0) ellipse (2.83cm and 0.3cm);
            \draw[-{Latex[length=5mm]}, very thick, myBlue] 
                plot[domain=180:330, samples=50, variable=\t] 
                ({2.83*cos(\t)}, {1 + 0.3*sin(\t)}, 0);
        
            \draw[red, dashed, thick] (A) -- (0, 2, 0);
        
            \node[circle, draw=none, fill=carbon, minimum size=35pt] at (A) {$C_{p_i}$};
            \node[circle, draw=none, fill=cyan, minimum size=35pt] at (top) {$C_{p_{i+1}}$};
            \node[circle, draw=none, fill=carbon, minimum size=35pt] at (bot) {$C_{p_{i-1}}$};
        
            \draw[-, dashed] (-0.94, 0.33, 0) 
                arc[start angle=160.5, end angle=270, radius=1cm];
            \node[anchor=south east] at (-1.06, -0.75, 0) {109.5$^\circ$};
        
            \begin{scope}[shift={(-4,-2,0)}, scale=1.2] 
                \def\rAxis{2} 
                \draw[->, thick, red!70!black] (0,0,0) -- (\rAxis,0,0) node[anchor=north east]{$x$};
                \draw[->, thick, red!70!black] (0,0,0) -- (0,\rAxis,0) node[anchor=north west]{$y$};
                \draw[->, thick, red!70!black] (0,0,0) -- (0,0,\rAxis) node[anchor=south east]{$z$};
            \end{scope}
        
        \end{tikzpicture}
        }
        \caption{Potential next positions of $C_{p_{i+1}}$ given by the circle in blue, when $C_{p_{i-1}}$ and $C_{p_{i}}$ are fixed.}
        \label{fig:pos_ajout}
    \end{subfigure}
    \begin{subfigure}[t]{0.51\textwidth}
        \centering
        \resizebox{0.5\linewidth}{!}{
        \tdplotsetmaincoords{-10}{0}
        
        \begin{tikzpicture}[tdplot_main_coords, scale=1]
            \definecolor{myBlue}{RGB}{98, 96, 255}
        
            \coordinate (A) at (0, 0, 0); 
            \coordinate (top) at (-2.83, 1, 0);
            \coordinate (bot) at (0, -3, 0);
        
            \coordinate (H1)  at ( 1.4143009,  1.0002470, -2.4493385);
            \coordinate (H2)  at ( 1.4143009,  1.0002470,  2.4493385);
        
            \draw[thick] (A) -- (bot);
            \draw[thick] (A) -- (top);
            \draw[thick] (A) -- (H1);
            \draw[thick] (A) -- (H2);
            
            \draw[red, dashed, thick] (A) -- (0, 2, 0);
        
            \node[circle, draw=none, fill=carbon, minimum size=35pt] at (A) {$C_{p_i}$};
            \node[circle, draw=none, fill=cyan, minimum size=35pt] at (top) {$C_{p_{i+1}}$};
            \node[circle, draw=none, fill=carbon, minimum size=35pt] at (bot) {$C_{p_{i-1}}$};
            \node[circle,fill=hydrogen,minimum size=18pt]  at (H1) {H};
            \node[circle,fill=hydrogen,minimum size=18pt]  at (H2) {H};
        
            \draw[-, dashed] (-0.94, 0.33, 0) 
                arc[start angle=160.5, end angle=270, radius=1cm];
            \node[anchor=south east] at (-1.06, -0.75, 0) {109.5$^\circ$};

            \draw pic[draw,"$109.5^\circ$",angle radius=0.8cm,angle eccentricity=1.5,dashed]{angle = bot--A--H2};
            \draw pic[draw,"$109.5^\circ$",angle radius=0.8cm,angle eccentricity=1.5,dashed]{angle = H1--A--top};
        
            \begin{scope}[shift={(-4,-2,0)}, scale=1.2] 
                \def\rAxis{2} 
                \draw[->, thick, red!70!black] (0,0,0) -- (\rAxis,0,0) node[anchor=north east]{$x$};
                \draw[->, thick, red!70!black] (0,0,0) -- (0,\rAxis,0) node[anchor=north west]{$y$};
                \draw[->, thick, red!70!black] (0,0,0) -- (0,0,\rAxis) node[anchor=south east]{$z$};
            \end{scope}
        
        \end{tikzpicture}
        }
        \caption{Full tetrahedral carbon $C_{p_{i}}$. $C_{p_i}$'s hydrogens positions are known when position of $C_{p_{i+1}}$ is fixed.}
        \label{fig:pos_ajout_with_hydro}
    \end{subfigure}

\caption{Next pattern positions and tetrahedral carbon construction schematic representation.}
\label{fig:pos_ajout_all}
\end{figure}
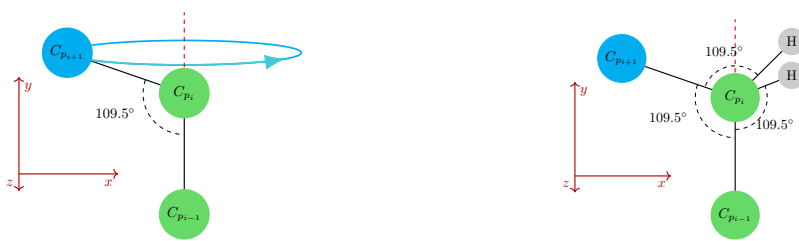

A first, naive approach consists in discretizing the placement circle and independently testing each candidate position for collisions with existing atoms. This results in a large number of queries to the \emph{Spherical Range Emptiness} (or \emph{Threshold Nearest Neighbor}) problem~\cite{bentley1975}, defined as follows: given a finite set $S \subseteq \mathbb{R}^3$, a query point $c \in \mathbb{R}^3$, and a threshold $d > 0$, decide whether $\forall s \in S,\ \|c - s\| \geq d$.

When the two additional hydrogen atoms are taken into account, the number of such queries becomes too large. Instead, we exploit the simple geometric structure of the feasible positions on the placement circle. Each existing atom $s \in S$ induces a forbidden angular interval on the circle corresponding to positions that violate the collision constraint. Moreover, the number of atoms that can generate such intervals is bounded by a constant, since only atoms lying within collision distance of the circle need to be considered. As a result, the set of valid positions can be represented by a bounded set of angular intervals.

For each atom, we compute and exclude the forbidden interval, yielding a set of admissible angular regions. The same construction is applied to the two hydrogen atoms, whose positions depend deterministically on the chosen angle $\theta$ via fixed angular offsets derived from the tetrahedral angle ($109.5^{\circ}$). The final set of valid placements is obtained by intersecting the admissible angular regions of all three atoms we place.

By storing the atoms of $S$ in an appropriate spatial data structure, we can identify the atoms close to the placement circle in constant time, since the number of atoms within a fixed radius is bounded by a constant. Since all subsequent operations involve only a bounded number of atoms and angular intervals, the overall complexity of this procedure is constant. Further details of this construction are provided in Appendix~\ref{sec:SchemaCollision-free}.

\subsection{Selecting the Best Positions}

The admissible angular intervals correspond to infinitely many feasible placements. To obtain a finite set of candidates, we discretize these intervals and guide the selection of angles using distance-based heuristics that estimate how promising a placement is. We consider three such heuristics, each inducing a different discretization strategy. The first relies on the Euclidean distance to the target, which is inexpensive to compute but ignores collision constraints and, therefore, provides only a coarse estimate. The second uses a distance computed on a discretized collision-aware grid, yielding a more accurate evaluation of the remaining path length but at a significantly higher computational cost. The third heuristic combines these two distances in order to balance efficiency and accuracy, leveraging the speed of the Euclidean metric only when relevant.

\subparagraph*{Euclidean distance}
The Euclidean distance from a candidate position to the target vertex $t$ is inexpensive to compute, but ignores obstacles. It performs well when $t$ is in direct line of sight, but may fail in cluttered environments where collisions prevent straight-line connections  as in Figure~\ref{fig:yarzun}.

\begin{figure}[b]
    \centering
    \includegraphics[width=0.2\linewidth]{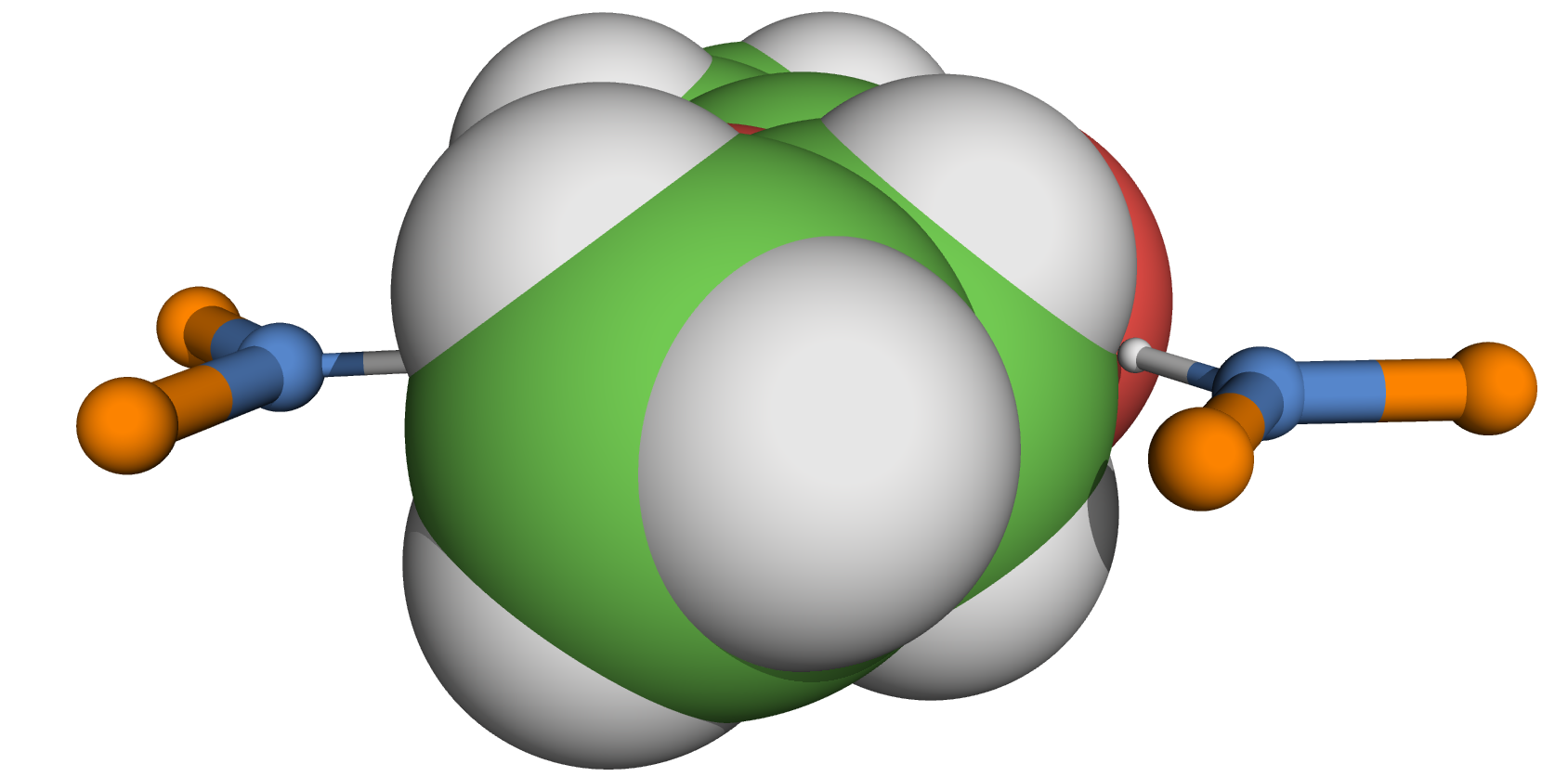}
    \caption{Ethyl propionate and its binding patterns. No molecular path can be found between the orange disconnected vertices when using the Euclidean distance.}
    \label{fig:yarzun}
\end{figure}

To generate candidate positions, we exploit the fact that $\theta^*$, the angle minimizing the Euclidean distance to $t$, can be computed analytically. Let $\Delta\theta$ denote the total angular width of the valid intervals and $n$ the desired number of angular samples. The angular step is $\Delta\theta / n$, bounded below by some minimum angular spacing. The first candidate angle is chosen as closely as possible to $\theta^*$ within the intervals of valid positions. The following positions are computed from the first by incrementing the angular step symmetrically on both sides, up to the limit imposed by the branching factor. If a sampled angle lies in a forbidden interval, it is moved to the closest valid angle in the sampling direction, without reversing direction.

\subparagraph*{Discretized Distance}
To better account for obstacles, we use a collision-aware distance computed on a discretized representation of space. Candidate positions are first selected within the valid angular intervals as follows: each interval initially contains a single uniformly placed sample, and additional samples are iteratively added to the interval whose current subdivision contains the longest segment. This adaptive refinement yields a well-distributed set of candidate positions.

We then evaluate each candidate by estimating its distance to the target using a voxel-based grid that excludes colliding cells. The voxel grid is constructed from the bounding box of the molecule, defined by the minimum and maximum coordinates along each axis and expanded by a margin of $d_{\text{weak}}$. The grid resolution is set to $0.05\,\mathrm{nm}$, and grid vertices that lie closer than a given threshold to any atom (modeled as spheres) are marked as unusable. Each free voxel corresponds to a node in a graph, and edges connect voxels in the 26-neighborhood, with weights equal to the Euclidean distance between voxel centers. The graph is augmented with the candidate positions and the target, each connected to their nearest grid neighbors. Shortest-path distances are computed using the $A^*$ algorithm~\cite{hart1968}. Since all candidates share the same target, we also evaluate the single-source multi-target variant ($SSMTA^*$)~\cite{htoo2013}.

\subparagraph*{Hybrid Distance} We also consider a hybrid heuristic that either uses the Euclidean distance when no obstacle blocks visibility from the current connection pattern to the target, and switches to the discretized distance otherwise. Visibility is determined by checking whether the line segment between the current connection pattern and the target intersects any atom, where atoms are modeled as spheres.

\subsection{Bounding the Search Space}\label{sec:Control_search_space}

\subparagraph{Branching Factor} To limit the combinatorial explosion of the visited partial molecular path, we limit the number of positions retained at each step by a \textbf{branching factor}. Only this number of positions with the best heuristic scores are kept and explored further. This balances completeness and efficiency, allowing us to find valid molecular paths while keeping computation times reasonable.

\subparagraph{Size Pruning}

To prevent unbounded exploration during molecular path construction, we introduce cut-offs that limit the number of patterns in a path. In its simplest form, this limit is a fixed maximum number of patterns defined a priori. Although this guarantees termination, it does not favor the shortest paths. To do so, we refine this limit by keeping track of the length (in number of patterns) of the shortest path found so far between the start and end vertices. Any newly constructed partial path whose length exceeds this bound is discarded. We refer to this rule as the \textbf{Min Length Cut}.

We further strengthen this criterion by estimating the minimum number of additional patterns required to complete the current partial path. This estimate is obtained from the Euclidean distance to the destination and the minimal number of patterns needed to span that distance. If the resulting projected path length exceeds the length of the shortest path found so far, the branch is pruned. We call this rule the \textbf{Projected Length Cut}. Together, these cuts effectively restrict the exploration on the shortest molecular paths.

%% file: section/interconnection.tex
We have shown how to generate a molecular path between two vertices. To construct a complete cage interacting with the substrate, all binding patterns must be interconnected to form a single connected structure. In this section, we introduce the notion of an \emph{interconnection tree}, which formalizes how binding patterns can be linked. We then present an efficient enumeration algorithm for interconnection trees. Combined with the molecular path construction described in the previous section, this approach enables the generation of complete molecular cages.

\subsection{Interconnection Tree}

We model the selection of binding patterns to connect as a graph-theoretic problem, abstracting away geometric and chemical constraints. We consider a \textbf{complete multipartite graph} \(G = (V,E)\), where the vertex set is partitioned into \(V_1,\dots,V_k\) and $E = \{(u,v) \mid i \neq j,\ u \in V_i,\ v \in V_j\}$. Each part \(V_i\) represents a binding pattern, and each vertex in \(V_i\) corresponds to an atom that can serve as an endpoint of a molecular path.

We define the function $\partition$ that maps each vertex $v$ to the part $V_i$ to which it belongs. This definition naturally extends to edges by setting $\partition((u,v)) = (\partition(u), \partition(v))$.
The \emph{quotient graph} of \(G\) (see \cite{bretto_elements_2022}) is then defined as
$
G_{\partition} = (\{V_i\}_{i \in [k]}, \partition(E)).
$

Our objective is to connect all parts in \(G_{\partition}\) while ensuring that each vertex of \(G\) is used at most once, since an atom can participate in only one molecular path. Moreover, we aim to minimize the number of connections, which leads to simpler structures and facilitates chemical synthesis. This leads to the following definition, where the symbol $\sqcup$ denotes disjoint union.

\begin{definition}
An \textbf{interconnection tree} of a multipartite graph
\(G = (V = V_1 \sqcup \dots \sqcup V_k, E)\)
is a set \(T \subseteq E\) such that \(T\) is a matching in \(G\) and
\(\partition(T)\) is a spanning tree of \(G_{\partition}\).
\end{definition}

Examples of interconnection trees are shown in Figure~\ref{fig:sol_motLiant_ACANIL}.

\begin{figure}[t]
    \centering
    \begin{subfigure}[t]{0.49\textwidth}
    \centering
    \resizebox{0.5\linewidth}{!}{
        \begin{tikzpicture}[scale=0.8]
        \node[draw, circle, color=blue, minimum size=2cm] (C1) at (0,4) {};
        \node[anchor=center, color=blue, color=blue] (C1label) at (0, 3) {$V_1$};
        \node[draw, circle, minimum size=0.8cm, inner sep=0pt] (1) at (-0.5, 4) {1};
        \node[draw, circle, minimum size=0.8cm, inner sep=0pt] (2) at (0.5, 4) {2};
    
        \node[draw, circle, color=blue, minimum size=2cm] (C2) at (5,4) {};
        \node[anchor=center, color=blue] (C2label) at (5, 3) {$V_2$};
        \node[draw, circle, minimum size=0.8cm, inner sep=0pt] (3) at (4.5, 4) {3};
        \node[draw, circle, minimum size=0.8cm, inner sep=0pt] (4) at (5.5, 4) {4};
    
        \node[draw, circle, color=blue, minimum size=3cm] (C3) at (0,0) {};
        \node[anchor=center, color=blue] (C3label) at (0, -1.6) {$V_3$};
        \node[draw, circle, minimum size=0.8cm, inner sep=0pt] (5) at (-0.5, 0) {5};
        \node[draw, circle, minimum size=0.8cm, inner sep=0pt] (6) at (-0.5, 1) {6};
        \node[draw, circle, minimum size=0.8cm, inner sep=0pt] (7) at (0.5, 1) {7};
        \node[draw, circle, minimum size=0.8cm, inner sep=0pt] (8) at (0.5, 0) {8};
        \node[draw, circle, minimum size=0.8cm, inner sep=0pt] (9) at (0.5, -1) {9};
        \node[draw, circle, minimum size=0.8cm, inner sep=0pt] (10) at (-0.5, -1) {10};
    
        \node[draw, circle, color=blue, minimum size=3cm] (C4) at (5,0) {};
        \node[anchor=center, color=blue] (C4label) at (5, -1.6) {$V_4$};
        \node[draw, circle, minimum size=0.8cm, inner sep=0pt] (11) at (4.5, 0) {11};
        \node[draw, circle, minimum size=0.8cm, inner sep=0pt] (12) at (4.5, 1) {12};
        \node[draw, circle, minimum size=0.8cm, inner sep=0pt] (13) at (5.5, 1) {13};
        \node[draw, circle, minimum size=0.8cm, inner sep=0pt] (14) at (5.5, 0) {14};
        \node[draw, circle, minimum size=0.8cm, inner sep=0pt] (15) at (5.5, -1) {15};
        \node[draw, circle, minimum size=0.8cm, inner sep=0pt] (16) at (4.5, -1) {16};

        \draw[color = red] (1) -- (6);
        \draw[color = red] (2) -- (3);
        \draw[color = red] (7) -- (11);
    
        \end{tikzpicture}
        }
        \label{fig:sol1_inter_motLiant_ACANIL}
    \end{subfigure}
    \begin{subfigure}[t]{0.49\textwidth}
    \centering
    \resizebox{0.5\linewidth}{!}{
        \begin{tikzpicture}[scale=0.8]
        \node[draw, circle, color=blue, minimum size=2cm] (C1) at (0,4) {};
        \node[anchor=center, color=blue, color=blue] (C1label) at (0, 3) {$V_1$};
        \node[draw, circle, minimum size=0.8cm, inner sep=0pt] (1) at (-0.5, 4) {1};
        \node[draw, circle, minimum size=0.8cm, inner sep=0pt] (2) at (0.5, 4) {2};
    
        \node[draw, circle, color=blue, minimum size=2cm] (C2) at (5,4) {};
        \node[anchor=center, color=blue] (C2label) at (5, 3) {$V_2$};
        \node[draw, circle, minimum size=0.8cm, inner sep=0pt] (3) at (4.5, 4) {3};
        \node[draw, circle, minimum size=0.8cm, inner sep=0pt] (4) at (5.5, 4) {4};
    
        \node[draw, circle, color=blue, minimum size=3cm] (C3) at (0,0) {};
        \node[anchor=center, color=blue] (C3label) at (0, -1.6) {$V_3$};
        \node[draw, circle, minimum size=0.8cm, inner sep=0pt] (5) at (-0.5, 0) {5};
        \node[draw, circle, minimum size=0.8cm, inner sep=0pt] (6) at (-0.5, 1) {6};
        \node[draw, circle, minimum size=0.8cm, inner sep=0pt] (7) at (0.5, 1) {7};
        \node[draw, circle, minimum size=0.8cm, inner sep=0pt] (8) at (0.5, 0) {8};
        \node[draw, circle, minimum size=0.8cm, inner sep=0pt] (9) at (0.5, -1) {9};
        \node[draw, circle, minimum size=0.8cm, inner sep=0pt] (10) at (-0.5, -1) {10};
    
        \node[draw, circle, color=blue, minimum size=3cm] (C4) at (5,0) {};
        \node[anchor=center, color=blue] (C4label) at (5, -1.6) {$V_4$};
        \node[draw, circle, minimum size=0.8cm, inner sep=0pt] (11) at (4.5, 0) {11};
        \node[draw, circle, minimum size=0.8cm, inner sep=0pt] (12) at (4.5, 1) {12};
        \node[draw, circle, minimum size=0.8cm, inner sep=0pt] (13) at (5.5, 1) {13};
        \node[draw, circle, minimum size=0.8cm, inner sep=0pt] (14) at (5.5, 0) {14};
        \node[draw, circle, minimum size=0.8cm, inner sep=0pt] (15) at (5.5, -1) {15};
        \node[draw, circle, minimum size=0.8cm, inner sep=0pt] (16) at (4.5, -1) {16};

        \draw[color = red] (2) -- (7);
        \draw[color = red] (3) -- (8);
        \draw[color = red] (4) -- (13);
        
        \end{tikzpicture}
        }
        \label{fig:sol2_inter_motLiant_ACANIL}
    \end{subfigure}
    \caption{Two interconnection trees in red over a complete multipartite graph with 4 parts in blue, representing the binding patterns of the acetanilide molecule.}
    \label{fig:sol_motLiant_ACANIL}
\end{figure}
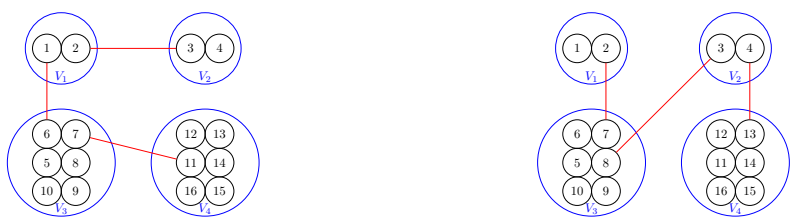

In general, deciding whether a multipartite graph admits an interconnection tree is \(\NP\)-hard, by reduction from the Hamiltonian path problem (see Appendix~\ref{appendix:interconnection_tree}). However, for complete multipartite graphs, existence admits a simple characterization.

\begin{lemma}[Proof in Appendix~\ref{appendix:interconnection_tree}] \label{lemma:complete}
Let \(G = (V = V_1 \sqcup \dots \sqcup V_k, E)\) be a complete multipartite graph.
Then \(G\) admits an interconnection tree if and only if
\(|V| \geq 2(k-1)\).
\end{lemma}

\subsection{Enumeration of Interconnection Trees}

Let \(\ICT(G)\) denote the set of interconnection trees of \(G\). We propose an efficient algorithm to enumerate all elements of \(\ICT(G)\). By efficient, we mean that the delay between the output of two consecutive solutions is bounded by the size of a solution, a standard complexity measure in enumeration problems~\cite{strozecki2021}.

A classical approach to enumeration relies on binary partitioning: solutions are divided according to whether they contain a given edge. Let \(G_{\setminus e}\) denote the graph obtained from \(G\) by removing an edge \(e\), and let \(G_{/e}\) denote the contraction of \(e = (u,v)\), where vertices \(u\) and \(v\) are identified and their corresponding parts are merged.

\begin{lemma}\label{lemma:contraction}
Let \(G\) be a multipartite graph. A set \(T\) is an interconnection tree of \(G\) containing an edge \(e\) if and only if \(T \setminus \{e\}\) is an interconnection tree of \(G_{/e}\).
\end{lemma}

Let \(S\) be a set of edge sets and let $e$ be an edge \(e\), we write
\(S \oplus e = \{ s \cup \{e\} \mid s \in S\}\).
By Lemma~\ref{lemma:contraction}, we can decompose $\ICT(G)$ into two disjoint sets as follows
$$
\ICT(G) = \ICT(G_{\setminus e}) \;\sqcup\; \ICT(G_{/e}) \oplus e
.$$ A similar decomposition is used for enumerating spanning trees~\cite{uno2015}, yielding an algorithm with linear delay and constant amortized delay. However, such algorithm requires efficiently deciding whether the recursive subproblems are non-empty in order to avoid useless recursion (a technique known as \emph{flashlight search}). In our setting, even if \(G\) is complete multipartite, the graph \(G_{\setminus e}\) is not, and deciding whether \(\ICT(G_{\setminus e})\) is empty becomes \(\NP\)-complete.

To address this issue, we refine the decomposition so that all recursive subproblems remain complete multipartite graphs, allowing us to apply Lemma~\ref{lemma:complete}. We fix an arbitrary ordering of the vertices of \(G\) that respects the partition: if \(u \in V_i\), \(v \in V_j\), and \(i < j\), then \(u < v\). For a vertex \(u\), we let \(G_{\geq u}\) denote the graph obtained from \(G\) by removing all vertices \(u' < u\). This leads to the following decomposition:
\begin{equation}
\ICT(G) = \bigsqcup_{\substack{(u,v) \in E \\ u \in V_1}} \ICT\bigl((G_{/(u,v)})_{\geq u}\bigr) \oplus (u,v).
\label{eq:large_union}
\end{equation}

Equation~\ref{eq:large_union} can be proved inductively. By induction hypothesis, we have the equation with an union over some initial segment $I$  of the edges and an additional term equal to $ICT(G_{\setminus I})$. We let $e$ be the first edge in the order after $I$ and we apply Lemma~\ref{lemma:contraction} to $ICT(G_{\setminus I})$ and the edge $e$ to prove the induction. 

Equation~\ref{eq:large_union} directly yields a recursive enumeration algorithm. With appropriate data structures, we can maintain in constant time whether each part is non-empty and whether the condition of Lemma~\ref{lemma:complete} is satisfied.

\begin{theorem}[Proof in Appendix~\ref{appendix:interconnection_tree}]
Let \(G\) be a complete multipartite graph with \(k\) parts. Then \(\ICT(G)\) can be enumerated with worst-case delay \(O(k)\) and amortized delay \(O(1)\).
\end{theorem}

In practice, the number of interconnection trees may be very large. To prioritize trees that are more likely to yield short molecular paths, we associate a weight to each edge, corresponding to the spatial distance between the two atoms. The weight of an interconnection tree is defined as the sum of the weights of its edges.

Ideally, one would enumerate interconnection trees in non-decreasing order of weight. However, computing a minimum-weight interconnection tree is already \(\NP\)-hard. When the number of interconnection trees remains moderate, a practical strategy, called \textbf{ordered}, is therefore to enumerate all trees, sort them by weight, and subsequently construct molecular paths. This approach is justified because we show in Section~\ref{sec:evalGlobalOnSub} that the cost of generating interconnection trees is several orders of magnitude smaller than that of computing molecular paths.

%% file: section/eval.tex
All experiments were conducted on an AMD Ryzen~7~PRO~250 CPU with 30\,GB of RAM. 
The code was compiled with GCC~13.3.0 using the \texttt{-O2} optimization level. Unless stated otherwise, all experiments use the same parameter settings. Pruning is set to \emph{Projected Length Cut}, the branching factor to~3, the number of angular samples to~12, the angular separation to~15$^{\circ}$, the distance function to \emph{Hybrid}, and the grid step to~$0.05\,\mathrm{nm}$. The maximum path length is fixed to~15, and at most $1{,}000{,}000$ solutions are generated.
The source code and datasets used in this work are available in our GitHub repository\footnote{https://github.com/NoeDemange/MolecularCagesGeneration}.

The initial substrates are extracted from the \emph{Cambridge Structural Database} (CSD)~\cite{CSD}, with the help of a chemist, we have selected 20 small organic molecules (fewer than 100 atoms) with varying characteristics in terms of size and shape. In what follows, we use only a representative subset of the molecules tested. Each experiment takes as input a substrate together with its associated binding patterns. Substrates are identified by their CSD Database Identifier in uppercase; when a lowercase letter is appended, it indicates that additional paths have been introduced so that only a single molecular path remains to be constructed.

To the best of our knowledge, there is currently no existing algorithmic approach for generating molecular cages from a given substrate, which prevents direct comparison with prior work. Our evaluation is therefore conducted on instances proposed by a chemist. While molecular cages have been experimentally constructed for some substrates, these structures are not generated computationally, and no standardized metrics or computational baselines are available to quantitatively compare their quality with our approach.

\subsection{Molecular Path Generation}

In this section, we evaluate the influence of the main parameters involved in the generation of a single molecular path.

\subparagraph*{Pruning}

We first compare the baseline algorithm with the two pruning strategies introduced in Section~\ref{sec:Control_search_space}, namely the \emph{Min Length Cut} and the \emph{Projected Length Cut}, in order to assess their impact on computational efficiency.

The results are reported in Table~\ref{tab:cut-and-branching}. As expected, both cuts significantly reduce the number of explored solutions by favoring shorter paths. The Projected Length Cut explores the fewest partial paths and consistently achieves the lowest running time, making it the most effective pruning strategy in practice.

\begin{table}[b]
\centering
\setlength{\tabcolsep}{1.2pt}
\tiny
\caption{Effect of pruning strategies (left table) and branching factor (right table) on path generation performance in time and quality (length and NRMSD).}
\label{tab:cut-and-branching}
\begin{subtable}[t]{0.49\textwidth}
\centering
\begin{tabular}{lllll}
\scriptsize\textbf{Instance} &
\scriptsize\textbf{Cut} &
\scriptsize\textbf{\#Paths} &
\scriptsize\textbf{\shortstack{\#Partial\\ paths}} &
\scriptsize\textbf{\shortstack{Time\\ (ms)}} \\
\hline
\T
\scriptsize\multirow{3}{*}{ABABELa}
 & \scriptsize Baseline   & \scriptsize 45 & \scriptsize 30\,839 & \scriptsize 145 \\
 & \scriptsize Min Length        & \scriptsize10 & \scriptsize7\,009  & \scriptsize 35  \\
 & \scriptsize Projected Length & \scriptsize10 & \scriptsize6\,225  & \scriptsize 22  \\
\hline
\T
\scriptsize\multirow{3}{*}{BAHSUYa}
 & \scriptsize Baseline   & \scriptsize 5 & \scriptsize 138 & \scriptsize 1.8 \\
 & \scriptsize Min Length        & \scriptsize 5 & \scriptsize 120 & \scriptsize 1.8 \\
 & \scriptsize Projected Length & \scriptsize 5 & \scriptsize 120 & \scriptsize 1.8 \\
\hline
\T
\scriptsize\multirow{3}{*}{YILLAGc}
 & \scriptsize Baseline   & \scriptsize 41 & \scriptsize 5\,662 & \scriptsize 47 \\
 & \scriptsize Min Length    & \scriptsize 10 & \scriptsize 1\,276 & \scriptsize 12 \\
 & \scriptsize Projected Length  & \scriptsize 10 & \scriptsize 1\,134 & \scriptsize 8  \\
\hline
\end{tabular}
\end{subtable}
\hfill
\begin{subtable}[t]{0.49\textwidth}
\centering
\begin{tabular}{llllll}
\scriptsize\textbf{Instance} &
\scriptsize\textbf{\#Pos.} &
\scriptsize\textbf{\shortstack{Path\\ length}} &
\scriptsize\textbf{\#Paths} &
\scriptsize\textbf{NRMSD} &
\scriptsize\textbf{\shortstack{Time\\ (ms)}}  \\
\hline
\T
\scriptsize \multirow{4}{*}{ABABELa}
 & \scriptsize 1 & \scriptsize \emph{N/A} & \scriptsize \emph{N/A} & \scriptsize \emph{N/A} & \scriptsize 2.5 \\
 & \scriptsize 2 & \scriptsize 10 & \scriptsize 1 & \scriptsize 0.87 & \scriptsize 6.3 \\
 & \scriptsize 3 & \scriptsize 10 & \scriptsize 10 & \scriptsize 0.53 & \scriptsize 22 \\
 & \scriptsize 4 & \scriptsize 10 & \scriptsize 67 & \scriptsize 0.33 & \scriptsize 99 \\
\hline
\T
\scriptsize \multirow{4}{*}{BAHSUYa}
 & \scriptsize 1 & \scriptsize \emph{N/A} & \scriptsize \emph{N/A} & \scriptsize \emph{N/A} & \scriptsize 2.3 \\
 & \scriptsize 2 & \scriptsize \emph{N/A} & \scriptsize \emph{N/A} & \scriptsize \emph{N/A} & \scriptsize 1.7 \\
 & \scriptsize 3 & \scriptsize 5 & \scriptsize 2 & \scriptsize 0.82 & \scriptsize 1.8 \\
 & \scriptsize 4 & \scriptsize 5 & \scriptsize 3 & \scriptsize 0.82 & \scriptsize 10 \\
\hline
\T
\scriptsize \multirow{4}{*}{YILLAGc}
 & \scriptsize 1 & \scriptsize \emph{N/A} & \scriptsize \emph{N/A} & \scriptsize \emph{N/A} & \scriptsize 1.8 \\
 & \scriptsize 2 & \scriptsize \emph{N/A} & \scriptsize \emph{N/A} & \scriptsize \emph{N/A} & \scriptsize 2.1 \\
 & \scriptsize 3 & \scriptsize 8 & \scriptsize 5 & \scriptsize 0.76 & \scriptsize 6.9 \\
 & \scriptsize 4 & \scriptsize 7 & \scriptsize 32 & \scriptsize 0.56 & \scriptsize 15 \\
\hline
\end{tabular}
\end{subtable}
\end{table}

\subparagraph*{Branching factor}

We then study the impact of the branching factor of the path generation algorithm by varying the number of retained candidate positions from 1 to 4. The results are summarized in Table~\ref{tab:cut-and-branching}.

When only one or two positions are retained, no valid solution is found, as the construction of paths lacks sufficient degrees of freedom to avoid obstacles and to respect the margins of the final attachment to $t$. Increasing the branching factor to four significantly increases both the number of solutions and the computation time. A branching factor of $3$ offers the best compromise, yielding a reasonable number of solutions, moderate computational cost, and good path quality in terms of NRMSD. Consequently, we retain three positions in the remainder of this work.

\subparagraph*{Number of angular samples} 

We next evaluate the influence of the number of angular samples used to discretize the angular intervals when selecting candidate positions. We test values of 24, 12, 8, and 6 samples; these values have been chosen by dividing $360^\circ$ by $15^\circ$, $30^\circ$, $45^\circ$ and $60^\circ$ respectively. The results are shown in Table~\ref{tab:sweep-comparison}. 

Increasing the number of samples does not systematically improve solution quality, while reducing it does not necessarily reduce the computation time. Using too many samples often yields candidate positions that are geometrically close and offer limited diversity. Conversely, too few samples reduce the accuracy of position selection and may lead to longer paths. Moreover, when the number of samples is less than $8$, it may be smaller than the number of valid angular intervals, and some intervals may contain no candidate position, further limiting exploration. In practice, we use 12 angular samples, although the choice between 8 and 12 remains inconclusive based on our experiments. A more extensive empirical study would be required to determine the optimal number of samples or to adapt this parameter dynamically to the input instance.

\subparagraph*{Minimum angular spacing} 

We then analyze the impact of the minimum angular spacing between candidate positions. The tested values are 5, 10, 15, and 20 degrees, with results reported in Table~\ref{tab:sweep-comparison}. 
The choice of angular spacing has a limited influence on performance. However, a separation of 15 degrees consistently provides a good compromise between path quality and computational cost across all tested instances.

\begin{table}[t]
\centering
\caption{Influence of discretization parameters on path generation performance. Number of considered positions (\#Pos.) left and  angular spacing right.}
\label{tab:sweep-comparison}
\begin{subtable}[t]{0.49\textwidth}
\centering
\setlength{\tabcolsep}{2pt}
\tiny
\label{tab:discretized-position}
\begin{tabular}{llllll}
\scriptsize\textbf{Instance} &
\scriptsize\textbf{\#Pos.} &
\scriptsize\textbf{\shortstack{Path\\ length}} &
\scriptsize\textbf{\#Paths} &
\scriptsize\textbf{NRMSD} &
\scriptsize\textbf{\shortstack{Time\\ (ms)}} \\
\hline
\T
\scriptsize \multirow{4}{*}{ABABELa}
 & \scriptsize 24 & \scriptsize 10 & \scriptsize 3  & \scriptsize 0.92 & \scriptsize 25 \\
 & \scriptsize 12 & \scriptsize 10 & \scriptsize 10 & \scriptsize 0.53 & \scriptsize 24 \\
 & \scriptsize 8  & \scriptsize 9  & \scriptsize 2  & \scriptsize 1.21 & \scriptsize 17 \\
 & \scriptsize 6  & \scriptsize 11 & \scriptsize 6  & \scriptsize 0.41 & \scriptsize 44 \\
\hline
\T
\scriptsize \multirow{4}{*}{BAHSUYa}
 & \scriptsize 24 & \scriptsize 5 & \scriptsize 1 & \scriptsize 0.82 & \scriptsize 2.6 \\
 & \scriptsize 12 & \scriptsize 5 & \scriptsize 2 & \scriptsize 0.82 & \scriptsize 2.1 \\
 & \scriptsize 8  & \scriptsize 6 & \scriptsize 1 & \scriptsize 0.80 & \scriptsize 3.0 \\
 & \scriptsize 6  & \scriptsize 7 & \scriptsize 6 & \scriptsize 0.74 & \scriptsize 5.1 \\
\hline
\T
\scriptsize \multirow{4}{*}{YILLAGc}
 & \scriptsize 24 & \scriptsize 8 & \scriptsize 5 & \scriptsize 0.24 & \scriptsize 5.8 \\
 & \scriptsize 12 & \scriptsize 8 & \scriptsize 5 & \scriptsize 0.76 & \scriptsize 8.2 \\
 & \scriptsize 8  & \scriptsize 7 & \scriptsize 6 & \scriptsize 0.45 & \scriptsize 4.9 \\
 & \scriptsize 6  & \scriptsize 7 & \scriptsize 1 & \scriptsize 1.16 & \scriptsize 5.2 \\
\hline
\end{tabular}
\end{subtable}
\hfill
\begin{subtable}[t]{0.5\textwidth}
\centering
\setlength{\tabcolsep}{1.6pt}
\tiny
\label{tab:threshold-angle}
\begin{tabular}{llllll}
\scriptsize\textbf{Instance} &
\scriptsize\textbf{\shortstack{Angular\\ spacing}} &
\scriptsize\textbf{\shortstack{Path\\ length}} &
\scriptsize\textbf{\#Paths} &
\scriptsize\textbf{NRMSD} &
\scriptsize\textbf{\shortstack{Time\\ (ms)}} \\
\hline
\T
\scriptsize \multirow{4}{*}{ABABELa}
 & \scriptsize 5  & \scriptsize \emph{N/A} & \scriptsize \emph{N/A} & \scriptsize \emph{N/A} & \scriptsize 2.4 \\
 & \scriptsize 10 & \scriptsize 10 & \scriptsize 12 & \scriptsize 0.44 & \scriptsize 26 \\
 & \scriptsize 15 & \scriptsize 10 & \scriptsize 10 & \scriptsize 0.53 & \scriptsize 22 \\
 & \scriptsize 20 & \scriptsize 10 & \scriptsize 9  & \scriptsize 0.34 & \scriptsize 21 \\
\hline
\T
\scriptsize \multirow{4}{*}{BAHSUYa}
 & \scriptsize 5  & \scriptsize 6 & \scriptsize 4 & \scriptsize 0.51 & \scriptsize 2.1 \\
 & \scriptsize 10 & \scriptsize 6 & \scriptsize 5 & \scriptsize 0.51 & \scriptsize 1.6 \\
 & \scriptsize 15 & \scriptsize 5 & \scriptsize 2 & \scriptsize 0.81 & \scriptsize 1.4 \\
 & \scriptsize 20 & \scriptsize 5 & \scriptsize 1 & \scriptsize 1.29 & \scriptsize 1.7 \\
\hline
\T
\scriptsize \multirow{4}{*}{YILLAGc}
 & \scriptsize 5  & \scriptsize 8 & \scriptsize 2 & \scriptsize 0.43 & \scriptsize 3.2 \\
 & \scriptsize 10 & \scriptsize 7 & \scriptsize 3 & \scriptsize 0.76 & \scriptsize 2.6 \\
 & \scriptsize 15 & \scriptsize 8 & \scriptsize 5 & \scriptsize 0.76 & \scriptsize 8.0 \\
 & \scriptsize 20 & \scriptsize 7 & \scriptsize 1 & \scriptsize 1.03 & \scriptsize 4.9 \\
\hline
\end{tabular}
\end{subtable}
\end{table}

\subparagraph*{Distance types} 
Finally, we evaluate the impact of the distance function used to guide path generation. We compare the three proposed strategies: the Euclidean distance, the discretized distance, and the Hybrid distance which combines both. As an alternative implementation of $A^*$ to compute the discretized distance, we also evaluate $SSMTA^*$, which computes the same distance but to all candidate positions at the same time. Since both methods return identical distances, we report their results in a single row, listing the execution times of both algorithms.

\begin{table}[t]
\centering
\setlength{\tabcolsep}{3pt}
\tiny
\caption{Comparison of distance functions in terms of path quality and computational cost. For the discretized distance, we report the runtime of $A^*$ followed by that of $SSMTA^*$.}
\label{tab:sweep-stats-type_distance}
\begin{tabular}{llccccc}
\scriptsize\textbf{Instance} &
\scriptsize\textbf{Distance} &
\scriptsize\textbf{\#Results} &
\scriptsize\textbf{Path length} &
\scriptsize\textbf{\#Paths} &
\scriptsize\textbf{NRMSD} &
\scriptsize\textbf{Time (ms)} \\
\hline
\T
\scriptsize \multirow{3}{*}{ABABELa}
 & \scriptsize A* / SSMTA* & \scriptsize 8 & \scriptsize 10 & \scriptsize 6 & \scriptsize 0.68 & \scriptsize 100 / 107 \\
 & \scriptsize Euclidean & \scriptsize 12 & \scriptsize 10 & \scriptsize 12 & \scriptsize 0.64 & \scriptsize 11 \\
 & \scriptsize HYBRID & \scriptsize 10 & \scriptsize 10 & \scriptsize 10 & \scriptsize 0.53 & \scriptsize 23 \\
\hline
\T
\scriptsize \multirow{3}{*}{BAHSUYa}
 & \scriptsize A* / SSMTA* & \scriptsize 7 & \scriptsize 5 & \scriptsize 2 & \scriptsize 0.90 & \scriptsize 10 / 9 \\
 & \scriptsize Euclidean & \scriptsize 2 & \scriptsize 5 & \scriptsize 2 & \scriptsize 0.79 & \scriptsize 0.32 \\
 & \scriptsize HYBRID & \scriptsize 5 & \scriptsize 5 & \scriptsize 2 & \scriptsize 0.82 & \scriptsize 1.9 \\
\hline
\T
\scriptsize \multirow{3}{*}{YARZUN03a}
 & \scriptsize A* / SSMTA* & \scriptsize 107 & \scriptsize 11 & \scriptsize 29 & \scriptsize 0.39 & \scriptsize 78 / 76 \\
 & \scriptsize Euclidean & \scriptsize 0 & \scriptsize NA & \scriptsize NA & \scriptsize NA & \scriptsize 0.24 \\
 & \scriptsize HYBRID & \scriptsize 93 & \scriptsize 11 & \scriptsize 26 & \scriptsize 0.15 & \scriptsize 29 \\
\hline
\T
\scriptsize \multirow{3}{*}{YILLAGc}
 & \scriptsize A* / SSMTA* & \scriptsize 8 & \scriptsize 7 & \scriptsize 2 & \scriptsize 0.89 & \scriptsize 13 / 14 \\
 & \scriptsize Euclidean & \scriptsize 1 & \scriptsize 5 & \scriptsize 1 & \scriptsize 0.93 & \scriptsize 0.18 \\
 & \scriptsize HYBRID & \scriptsize 10 & \scriptsize 8 & \scriptsize 5 & \scriptsize 0.76 & \scriptsize 8.3 \\
\hline
\end{tabular}
\end{table}

As shown in Table~\ref{tab:sweep-stats-type_distance}, the Euclidean distance is by far the fastest to compute, but it fails to consistently produce valid paths, particularly in the presence of obstacles, as illustrated by the absence of results for instance YARZUN03a (shown in Figure~\ref{fig:yarzun}). In contrast, both $A^*$ and $SSMTA^*$ are significantly more robust, at the cost of substantially higher computation times. Despite its theoretical advantages, $SSMTA^*$ does not yield a noticeable speed-up over $A^*$ in our experiments. 

The Hybrid distance offers the best compromise between robustness and efficiency. It successfully generates valid paths in most instances while remaining considerably faster than purely graph-based distances. For this reason, we adopt the Hybrid distance in the remainder of this work.

\subsection{Interconnection Tree Enumeration}

We evaluate the practical performance of our interconnection tree enumeration algorithm on several instances corresponding to realistic molecular cages. Each instance is modeled as a complete multipartite graph with $k$ parts and $l$ vertices per part, and is denoted by $(k, l)$ in the experiments. The chosen values of $k$ and $l$ reflect the sizes encountered in our cage construction process.

The goal of this experiment is to assess the efficiency of the enumeration algorithm in terms of total runtime, amortized delay between successive solutions, and the overhead induced by storing interconnection trees and sorting them by weight. In our implementation, we adopt simpler data structures, which do not ensure constant amortized delay, but lead to improved practical delays on the small instances considered here.

\begin{table}[b]
\centering
\setlength{\tabcolsep}{3pt}
\tiny
\caption{Performance of the interconnection tree enumeration algorithm for several realistic instances. Instances exceeding a limit of two minutes are marked as \emph{Timeout (TO)}. }
\label{tab:execution-results-instances-transposed}
\begin{tabular}{lccccc}
\scriptsize\textbf{Instance} &
\scriptsize\textbf{\#Trees} &
\scriptsize\textbf{Time (ms)} &
\scriptsize\textbf{Delay (ns)} &
\scriptsize\textbf{Sorting overhead (\%)} &
\scriptsize\textbf{Storage overhead (\%)}\\
\hline
\T
\scriptsize (3,3)  & \scriptsize 162            & \scriptsize 0.006        & \scriptsize 36 & \scriptsize 116 & \scriptsize 66 \\
\scriptsize (3,6) & \scriptsize 3\,240         & \scriptsize 0.07        & \scriptsize 22 & \scriptsize 294 & \scriptsize 190 \\ 
\scriptsize (5,3) & \scriptsize 174\,960       &\scriptsize 5.8        & \scriptsize 34 & \scriptsize 255 & \scriptsize 164 \\ 
\scriptsize (5,6) & \scriptsize 107\,308\,800  & \scriptsize 2\,685   & \scriptsize 25 & \scriptsize 402 & \scriptsize 188 \\
\scriptsize (7,3) & \scriptsize 525\,404\,880  & \scriptsize 16\,742  & \scriptsize 32 & \scriptsize \emph{N/A} & \scriptsize \emph{N/A} \\
\scriptsize (7,6) & \scriptsize 4\,119\,428\,086 & \scriptsize \emph{TO} & \scriptsize 29 & \scriptsize \emph{N/A} & \scriptsize \emph{N/A} \\ 
\hline
\end{tabular}
\end{table}

The results reported in Table~\ref{tab:execution-results-instances-transposed} confirm that the amortized delay is very small and effectively constant in practice. As expected, storing and sorting interconnection trees is more expensive than generating them, and the relative cost of sorting increases with the number of trees. Nevertheless, the time spent per interconnection tree remains negligible compared to the time required to compute the molecular paths associated with a given tree.
Therefore, when the number of interconnection trees is reasonable, typically below $10^8$, we can afford to enumerate all trees and sort them by weight.

\subsection{Full Molecular Cage Generation}

In this section, we evaluate the full pipeline on real substrates, where multiple molecular paths must be generated to connect the selected binding patterns. We evaluate the early edge removal heuristic and two generation strategies and, independently of these choices, we analyze the resulting generated cages. Figure~\ref{fig:results_cages} illustrates cages produced by our approach, visualized using PyMOL~\cite{PyMOL}. The molecular cages behave as expected: binding sites are positioned to interact with the substrate, and the connecting paths define a shape complementary to that of the substrate. No path obstructs the internal cavity intended to host the substrate. 

\begin{figure}[t]
    \centering
    \begin{subfigure}[b]{0.32\textwidth}
        \centering
        \includegraphics[width=0.65\textwidth]{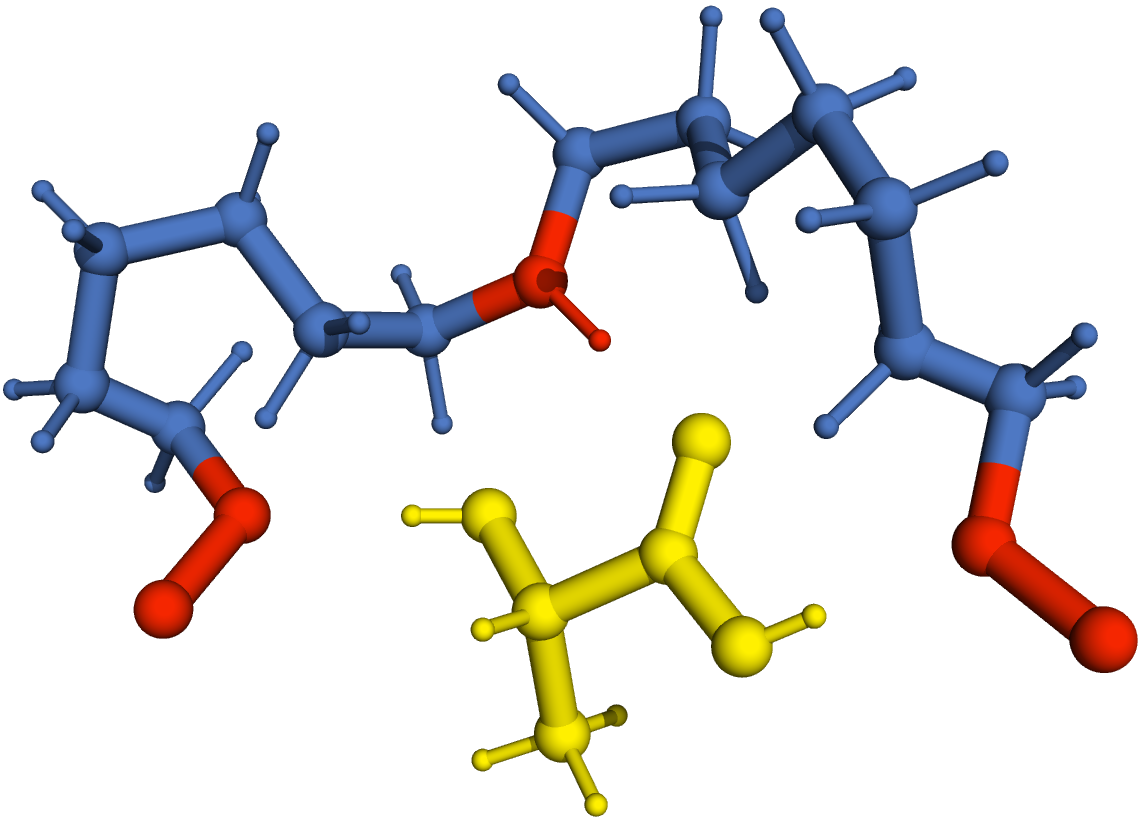}
        \caption{YILLAG: L-(+)-lactic acid}
        \label{fig:results_yillag}
    \end{subfigure}
    \begin{subfigure}[b]{0.32\textwidth}
        \centering
        \includegraphics[width=0.65\textwidth]{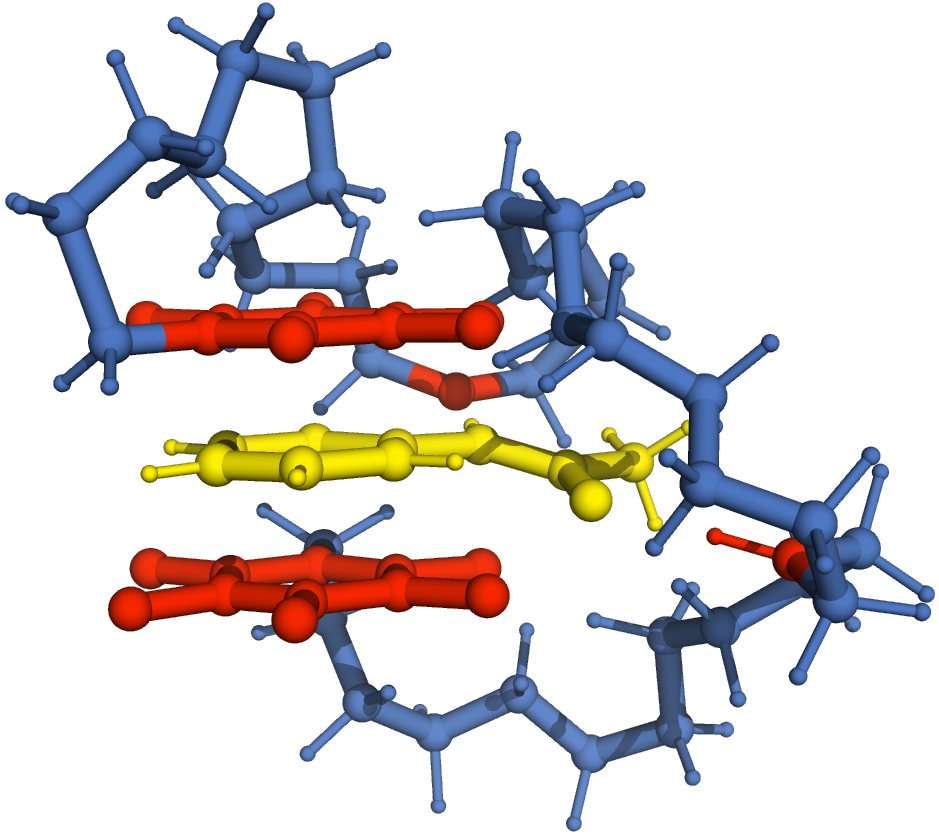}
        \caption{ACANIL01: Acetanilide}
        \label{fig:results_acanil01}
    \end{subfigure}
    \begin{subfigure}[b]{0.32\textwidth}
        \centering
        \includegraphics[width=0.65\textwidth]{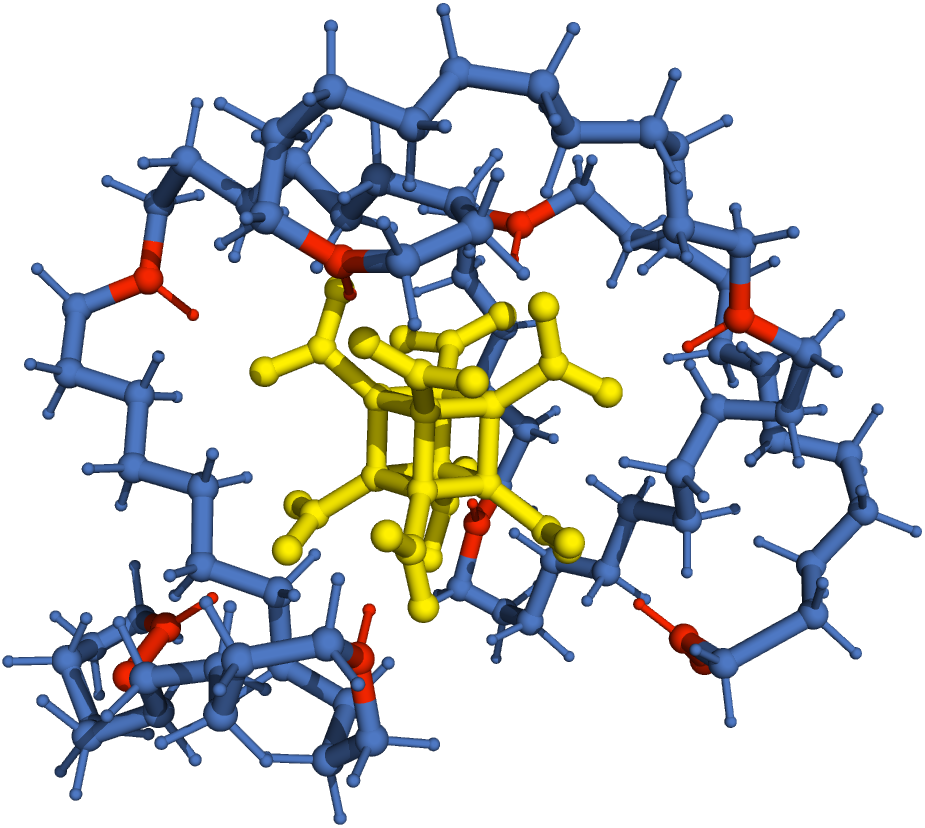}
        \caption{CUGDIR: Octanitrocubane}
        \label{fig:results_cugdir}
    \end{subfigure}\\
    \begin{subfigure}[b]{0.32\textwidth}
        \centering
        \includegraphics[width=0.8\textwidth]{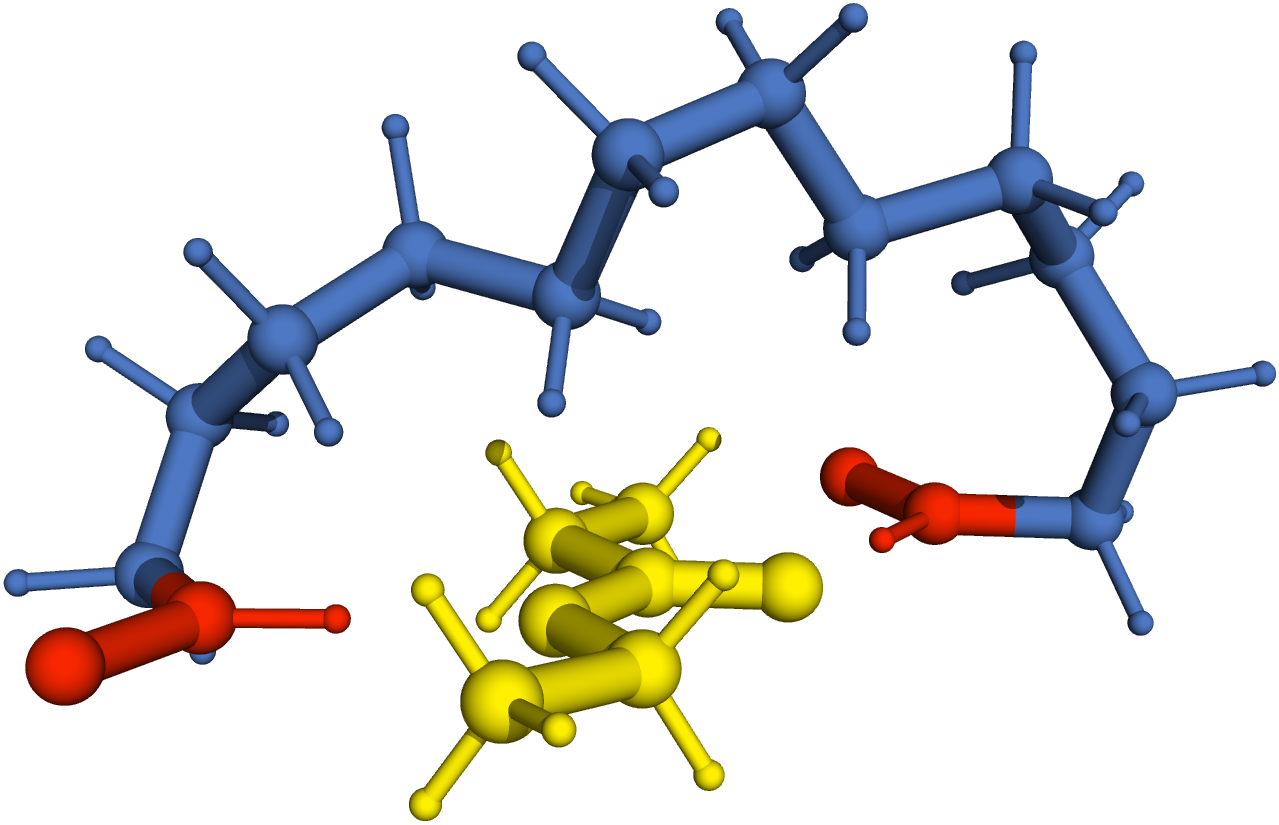}
        \caption{YARZUN03: Ethyl propionate}
        \label{fig:results_yarzun}
    \end{subfigure}
    \begin{subfigure}[b]{0.32\textwidth}
        \centering
        \includegraphics[width=0.8\textwidth]{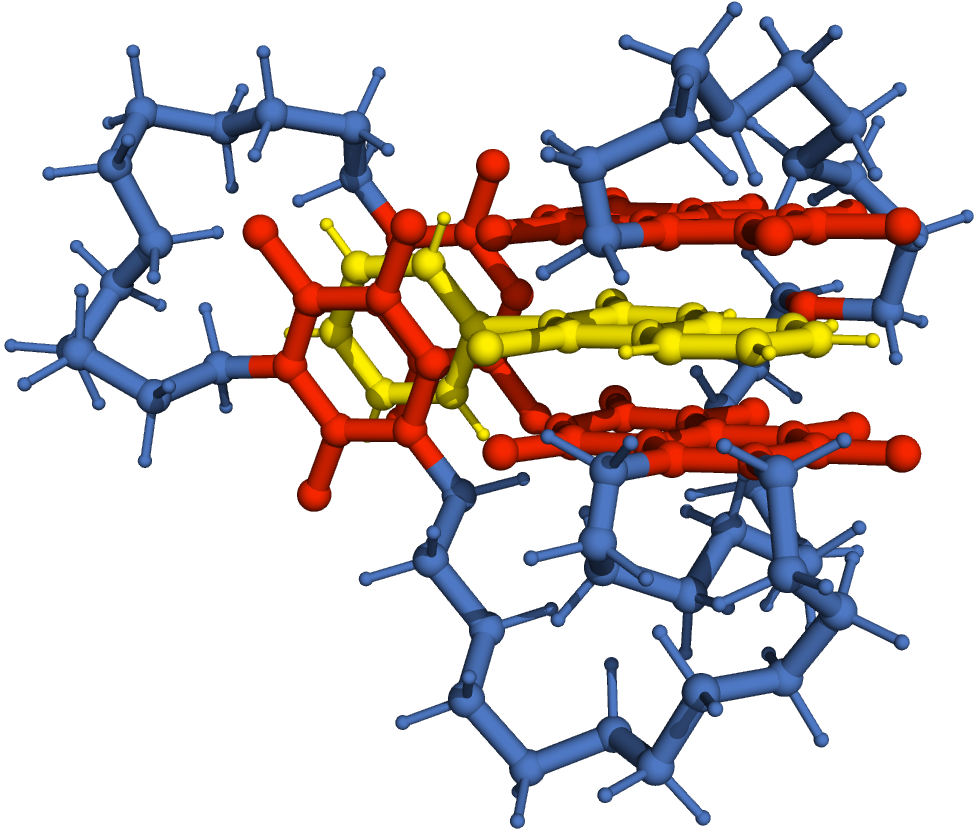}
        \caption{ABABEL: $4-\text{Chloro}-3-\text{phenylquinolin}-2(1H)-\text{one}$}
        \label{fig:results_ababel}
    \end{subfigure}
    \begin{subfigure}[b]{0.32\textwidth}
        \centering
        \includegraphics[width=0.8\textwidth]{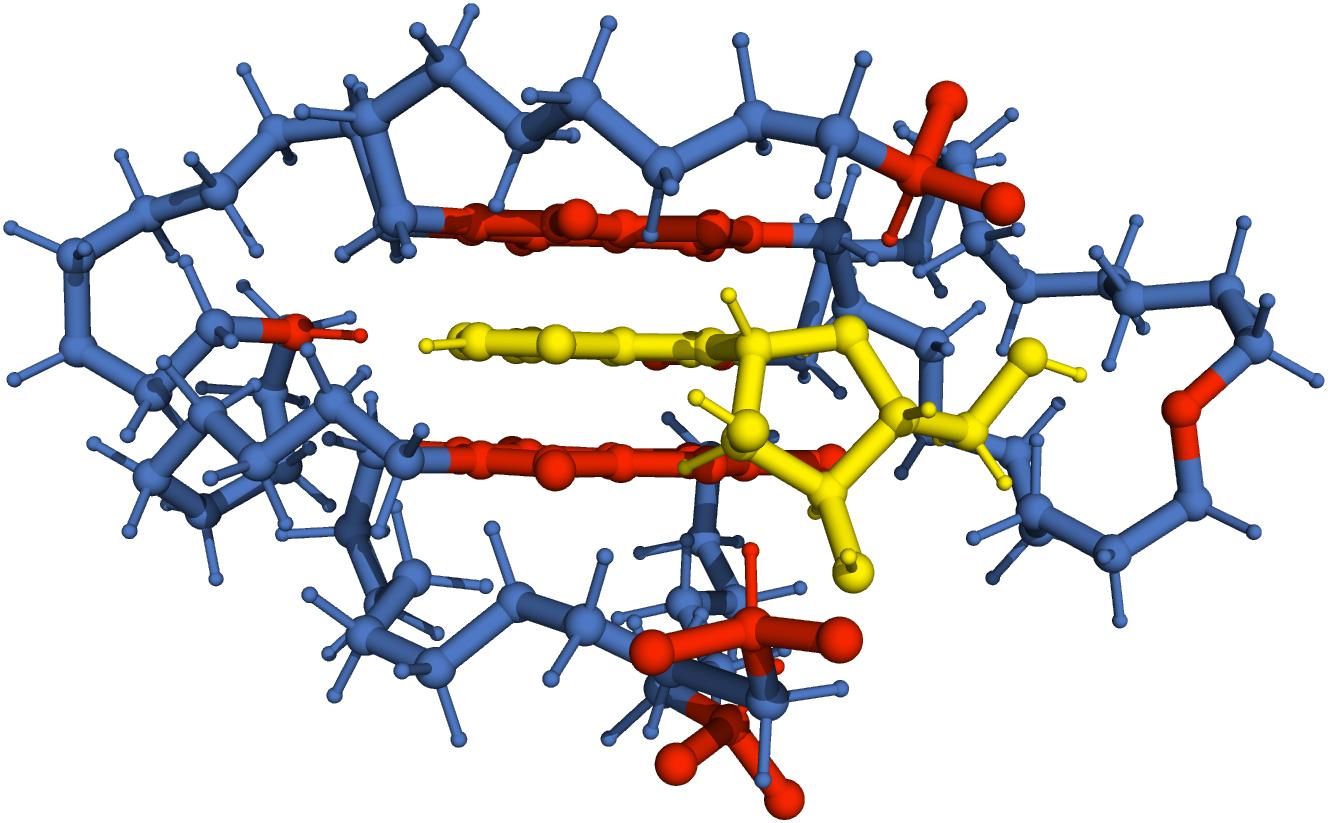}
        \caption{ADENOS10: Adenosine }
        \label{fig:results_adenos}
    \end{subfigure}
    \caption{Generated cages (blue) for three substrates (yellow) with their binding patterns (red).}
    \label{fig:results_cages}
\end{figure}

\subparagraph*{Early edge removal heuristic}
We introduce the \textbf{Early Edge Removal} heuristic to accelerate the exploration of interconnection trees. When the algorithm fails to construct a molecular path for a given edge, all remaining interconnection trees containing this edge are skipped. This pruning strategy eliminates large portions of the search space that are unlikely to yield valid or compact cages.

\begin{table}[b]
\centering
\setlength{\tabcolsep}{2pt}
\tiny
\caption{Impact of heuristics (left) and generation strategies (right) on solution and runtime. MNoA is for Minimum Number of Atoms, OTF for On-the-fly, ORD for ordered and \emph{TO} for Timeout.}
\label{tab:heuristic-and-mode}
\begin{subtable}[t]{0.5\textwidth}
\centering
\begin{tabular}{llllllll}
\scriptsize\textbf{Instance} &
\rot{Mode} &
\rot{\shortstack{Edge\\ Removal}} &
\rot{\#Cages} &
\rot{\#Trees} &
\rot{\shortstack{Average\\ NRMSD}} &
\rot{MNoA} &
\rot{Total (s)} \\
\hline
\T
\scriptsize \multirow{4}{*}{ ABCLUA10}
 & \scriptsize OTF & \scriptsize no  & \scriptsize 10 & \scriptsize 1  & \scriptsize 0.87 & \scriptsize 203 & \scriptsize 0.04 \\
 & \scriptsize OTF & \scriptsize yes & \scriptsize 10 & \scriptsize 1  & \scriptsize 0.87 & \scriptsize 203 & \scriptsize 0.04 \\
 & \scriptsize ORD & \scriptsize no  & \scriptsize 10 & \scriptsize 1  & \scriptsize 0.95 & \scriptsize 101 & \scriptsize 11.6 \\
 & \scriptsize ORD & \scriptsize yes & \scriptsize 10 & \scriptsize 1  & \scriptsize 0.95 & \scriptsize 101 & \scriptsize 11.5\\
\hline
\T
\scriptsize \multirow{4}{*}{ ABINOS}
 & \scriptsize OTF & \scriptsize no  & \scriptsize 10 & \scriptsize 166 & \scriptsize 0.63 & \scriptsize 151 & \scriptsize 15 \\
 & \scriptsize OTF & \scriptsize yes & \scriptsize 10 & \scriptsize 14  & \scriptsize 0.59 & \scriptsize 133 & \scriptsize 18 \\
 & \scriptsize ORD & \scriptsize no  & \scriptsize 10 & \scriptsize 4  & \scriptsize 0.99 & \scriptsize 91 & \scriptsize 5.9 \\
 & \scriptsize ORD & \scriptsize yes & \scriptsize 10 & \scriptsize 3  & \scriptsize 0.99 & \scriptsize 91 & \scriptsize 5.8 \\
\hline
\T
\scriptsize \multirow{4}{*}{ ACANIL01}
 & \scriptsize OTF & \scriptsize no  & \scriptsize 0  & \scriptsize 33 & \scriptsize \emph{N/A} & \scriptsize \emph{N/A} & \scriptsize \emph{TO} \\
 & \scriptsize OTF & \scriptsize yes & \scriptsize 10 & \scriptsize 11 & \scriptsize 0.87 & \scriptsize 99 & \scriptsize 5.4 \\
 & \scriptsize ORD & \scriptsize no  & \scriptsize 10 & \scriptsize 500  & \scriptsize 0.70 & \scriptsize 87 & \scriptsize 14 \\
 & \scriptsize ORD & \scriptsize yes & \scriptsize 10 & \scriptsize 14  & \scriptsize 0.70 & \scriptsize 87 & \scriptsize 3 \\
\hline
\end{tabular}
\end{subtable}
\hfill
\begin{subtable}[t]{0.48\textwidth}
\centering
\begin{tabular}{lllllll}
\scriptsize\textbf{Instance} &
\rot{Mode} &
\rot{\#Trees} &
\rot{\#Cages} &
\rot{\shortstack{Average\\ NRMSD}} &
\rot{MNoA} &
\rot{Total (s)} \\
\hline
\T
\scriptsize \multirow{2}{*}{ ABCLUA10}
 & \scriptsize OTF & \scriptsize 1      & \scriptsize 17\,424 & \scriptsize 0.67 & \scriptsize 191 & \scriptsize \emph{TO} \\
 & \scriptsize ORD    & \scriptsize 34 & \scriptsize 15\,964 & \scriptsize 0.74 & \scriptsize 95  & \scriptsize \emph{TO} \\
\hline
\T
\scriptsize \multirow{2}{*}{ ABINOS}
 & \scriptsize OTF & \scriptsize 66     & \scriptsize 8\,219 & \scriptsize 0.43 & \scriptsize 85 & \scriptsize \emph{TO} \\
 & \scriptsize ORD    & \scriptsize 174 & \scriptsize 6\,834 & \scriptsize 0.61 & \scriptsize 61 & \scriptsize \emph{TO} \\
\hline
\T
\scriptsize \multirow{2}{*}{ ACANIL01}
 & \scriptsize OTF & \scriptsize 50     & \scriptsize 564  & \scriptsize 0.57 & \scriptsize 90 & \scriptsize \emph{TO} \\
 & \scriptsize ORD & \scriptsize 264 & \scriptsize 2\,253 & \scriptsize 0.52 & \scriptsize 69 & \scriptsize \emph{TO} \\
\hline
\T
\scriptsize \multirow{2}{*}{ ALFUCO}
 & \scriptsize OTF & \scriptsize 26     & \scriptsize 280  & \scriptsize 0.64 & \scriptsize 100 & \scriptsize \emph{TO} \\
 & \scriptsize ORD   & \scriptsize 63 & \scriptsize 2\,528 & \scriptsize 0.43 & \scriptsize 61  & \scriptsize \emph{TO} \\
\hline
\T
\scriptsize \multirow{2}{*}{ BAHSUY}
 & \scriptsize OTF & \scriptsize 6 & \scriptsize 156 & \scriptsize 0.17 & \scriptsize 19 & \scriptsize 0.06 \\
 & \scriptsize ORD    & \scriptsize 6 & \scriptsize 156 & \scriptsize 0.17 & \scriptsize 19 & \scriptsize 0.05 \\
\hline
\end{tabular}
\end{subtable}

\end{table}

Table~\ref{tab:heuristic-and-mode} compares generation modes with and without this heuristic on real instances, using a time limit of 120 seconds and a maximum of 10 returned solutions. The heuristic proves highly effective: it prunes interconnection trees that would either fail entirely or lead to long paths, enabling the discovery of significantly smaller cages in substantially less time.

\subparagraph*{Generation strategies}
Finally, we compare two strategies for interconnection tree generation: \textbf{on-the-fly} and \textbf{ordered}. As shown in Table~\ref{tab:heuristic-and-mode}, when the time budget does not allow to try to construct paths for all trees enumerated, processing interconnection trees in increasing weight order yields smaller cages than the on-the-fly approach. When sufficient time is available to enumerate all trees and try to construct paths on each, both strategies find the same solutions, but the cost of sorting is negligible compared to the cost of molecular path construction.

%% file: section/conclusion.tex
In this paper, we introduced a computational framework to generate molecular cages tailored to a target substrate. By combining the placement of binding patterns, the construction of short molecular paths, and the enumeration of interconnection trees, our approach enables the automatic generation of chemically realistic cages while promoting geometric complementarity and substrate specificity. Extensive experimental results demonstrate that the proposed heuristics, in particular hybrid distance estimations and pruning strategies on interconnection trees, allow the construction of compact cages within reasonable computation times.

Several directions for future work emerge from this study. First, distance computation could be further improved by adapting algorithms and the spatial data structures to our specific problem, in order to obtain distance estimates that are both more accurate and faster to compute. Second,  methods to rigidify the cage could be explored, for instance by completing the structure  to reduce the length of molecular paths, leading to more stable and synthetically accessible architectures. Finally, allowing multiple binding patterns to be connected to a single attachment point could help further reduce molecular path lengths and increase the compactness of the generated cages.

%% file: section/appendix.tex
\section{Collision-Free Angular Intervals on a Circle problem}\label{sec:SchemaCollision-free}
To construct the circle, we define a center point $C_c$ from the positions
$C_{p_i}$ and $C_{p_{i-1}}$. The next vertex ($C_{p_{i+1}}$) is constrained to lie on a circle of radius
$R$ in a plane spanned by two orthonormal vectors $\vec{vec}_1$ and $\vec{vec}_2$.
Thus, any candidate position can be written as
\[
\mathcal{C}(\theta)
= \vec{C_c}
+ R \cos(\theta)\,\vec{vec}_1
+ R \sin(\theta)\,\vec{vec}_2 ,
\qquad \theta \in [0,2\pi].
\]
For each point $s_i$ in the set $S$ of positions of atoms, the collision constraint
$\|\mathcal{C}(\theta) - s_i\| < d$ defines one or more forbidden angular intervals. The valid placement domain for the next vertex is obtained by subtracting each forbidden interval to the full domain.

As previously discussed, once the next atom is positioned, the two missing hydrogen atoms can be added to complete the local bonding pattern. However, the position of the next vertex is not known in advance, but we know that an hydrogen atom satisfies the following constraint: the angle formed by the next vertex, the starting vertex $C_{p_i}$, and the hydrogen atom is  $109.5^\circ$.

Following the same geometric construction used for the next atom, we determine a circle on which each hydrogen atom must lie. This circle has center $C_{ch}$, radius $r_h$, and lies in a plane spanned by two orthonormal vectors $\vec{vec}_1$ and $\vec{vec}_2$. Any feasible position of a hydrogen atom can therefore be parameterized by an
angle $\phi \in [0,2\pi]$ as
\[
\mathcal{C}_h(\phi)
= \vec{C_{ch}}
+ r_h \cos(\phi)\,\vec{vec}_1
+ r_h \sin(\phi)\,\vec{vec}_2 ,
\qquad \phi \in [0,2\pi].
\]
Thus, we can also compute the valid angular intervals for the hydrogen atoms.
Since $R \neq r_h$, $\phi = \theta \pm 109.5^\circ$ cannot be directly used. Instead, the angular shift must be determined using the formula for the angle between two vectors:
\(cos(109.5^\circ)=\frac{\overrightarrow{a}.\overrightarrow{b}}{\lVert\overrightarrow{a}\rVert \times \lVert\overrightarrow{b}\rVert}\).

From this relation, we compute a value $\Delta$ giving the geometric correspondence between
the angular parameters of the two circles and we write $\phi = \theta \pm \Delta$. 
Thus, we obtain the valid angular intervals for all three
positions (the next atom and the two hydrogens), as shown in  Figure~\ref{fig:geometry_grid_positive_hydro}, and we intersect them to find the final set of $\theta$ which are valid placement region.

The number of intervals created by possible collisions can easily be bounded as follows.
All atoms in collision with an atom whose center lies on the circle are included in a region at distance at most $3col/2$ of this circle. 
Hence, the region containing the potential interfering atoms is of volume $2\pi R \pi (3col/2)^2 $
and each atom is of volume $(4/3)\pi col^3 $. Hence, the number of such atoms is bounded by the quotient of these values, $(27/8)\pi R/col \leq 14 $. Then, taking the intersection of the three set of intervals, multiplies this number by at most three. This is a crude bound and in practice, there are less than $10$ intervals.

\newcommand{\drawSphere}[4]{
    \shade[ball color=carbon] (#1, #2, #3) circle (0.15cm); 
    \node[anchor=south west] at (#1, #2, #3) {#4};          
}

\newcommand{\drawSphereHydro}[4]{
    \shade[ball color=hydrogen] (#1, #2, #3) circle (0.15cm); 
    \node[anchor=south west] at (#1, #2, #3) {#4};          
}

\tdplotsetmaincoords{75}{115}  
\begin{figure}[t]
    \centering
    \resizebox{0.5\linewidth}{!}{
    \begin{tikzpicture}[tdplot_main_coords,scale=2]    
    
            \coordinate (A) at (0, 0, 0);
            \coordinate (B) at (6, 0, 0);
            \coordinate (C) at (6, 6, 0);
            \coordinate (D) at (0, 6, 0);
            \coordinate (E) at (0, 0, 6);
            \coordinate (F) at (6, 0, 6);
            \coordinate (G) at (6, 6, 6);
            \coordinate (H) at (0, 6, 6);
    
            \foreach \i in {0,1,2} {
                \draw[gridline,dashed] (A) ++(0,0,2*\i) -- +(6,0,0);
                \draw[gridline,dashed] (A) ++(2*\i,0,0) -- +(0,0,6);
                \draw[gridline,dashed] (A) ++(2*\i,0,0) -- +(0,6,0);
                \draw[gridline,dashed] (A) ++(0,2*\i,0) -- +(6,0,0);
            }
            
            \draw[cubeedge,dashed] (A) -- (B) ;
            \draw[cubeedge] (B) -- (C);
            \draw[cubeedge] (C) -- (D);
            \draw[cubeedge,dashed] (A) -- (D);
            \draw[cubeedge] (E) -- (F) -- (G) -- (H) -- cycle;
            \draw[cubeedge,dashed] (A) -- (E);
            \draw[cubeedge] (B) -- (F);
            \draw[cubeedge] (D) -- (H);
    
            \draw[->,dashed] (A) -- (B) node[anchor=north east] {$x$};
            \foreach \i in {0,1,2,3} {
                \draw (A) ++(2*\i,0,0) -- ++(0,0,0.1) node[anchor=south east] {\footnotesize \i};
            }

            \draw[->,dashed] (A) -- (D) node[anchor=south west] {$y$};
            \foreach \i in {1,2,3} {
                \draw (A) ++(0,2*\i,0) -- ++(0.2,0,0) node[anchor= north west] {\footnotesize \i};
            }
    
            \draw[->,dashed] (A) -- (E) node[anchor=south] {$z$};
            \foreach \i in {1,2,3} {
                \draw (A) ++(0,0,2*\i) -- ++(0.2,0.1,0) node[anchor=north west] {\footnotesize \i};
            }

            \coordinate (Cs) at (3,3,2);       
            \coordinate (Cn) at (3,0,2);    
            \coordinate (Next) at (3, 3+0.998, 2+2.828448); 
            \coordinate (Centre) at (3, 3+0.998, 2); 
            \coordinate (Ch) at (3 ,3+0.73328, 2); 
            \coordinate (H1) at (4.796561 ,3+0.73328, 0.96373677); 
            \coordinate (H2) at (1.203438 ,3+0.73328, 0.96373677); 
    
            \draw[thick] (Cn) -- (Cs) node[midway, above left] {$0.15\,\mathrm{nm}$}; 
            \draw[thick] (Cs) -- (Next) node[midway, left] {$0.15\,\mathrm{nm}$}; 
            \draw[thick] (Cs) -- (H1) node[midway, left] {$0.11\,\mathrm{nm}$}; 
            \draw[thick] (Cs) -- (H2); 
            \draw[thick,dashed] (Cs) -- (Centre) node[midway, below] {};
            \draw[thick,dashed] (Next) -- (Centre) node[midway, right] {$R=0.1414\,\mathrm{nm}$}; 
            \draw[thick,dashed] (Ch) -- (H1) node[midway, right] {$r_h=0.1037\,\mathrm{nm}$}; 
            \draw[thick,dashed] (Ch) -- (H2); 
    
            \drawSphere{3}{3}{2}{$C_{p_i}$}; 
            \drawSphere{3}{0}{2}{$C_{p_{i-1}}$}; 
            \shade[ball color=carbon] (3, 3+0.998, 2+2.828448) circle (0.15cm); \node[anchor=south east] at (3, 3+0.998, 2+2.828448) {$C_{p_{i+1}}$};
            \filldraw[carbon] (3,3+0.998,2) circle[radius=1pt] ; \node[anchor=west] at (3, 3+0.998, 2)  {$C_{c}$}; 
            \filldraw[hydrogen] (3 ,3+0.73328, 2) circle[radius=1pt] ; \node[anchor=south] at (3, 3+0.73328, 2)  {$C_{ch}$}; 
            \shade[ball color=hydrogen] (4.796561 ,3+0.73328, 0.96373677) circle (0.15cm); \node[anchor=south west] at (4.796561 ,3+0.73328, 0.96373677)  {$H_1$};
            \shade[ball color=hydrogen] (1.203438 ,3+0.73328, 0.96373677) circle (0.15cm); \node[anchor=south west] at (1.203438 ,3+0.73328, 0.96373677)  {$H_2$};
    
            \tdplotsetthetaplanecoords{90}
            \tdplotdrawarc[tdplot_rotated_coords,-,blue]{(Cs)}{0.4}{379.47}{270}{anchor=south east}{$109.47^\circ$} 
            \tdplotdrawarc[tdplot_rotated_coords,-,blue]{(Next)}{0.5}{180}{199.47}{anchor=north}{$19.47^\circ$} 

            \tdplotsetthetaplanecoords{180}
            \tdplotdrawarc[tdplot_rotated_coords,->,red]{(Centre)}{2.828448}{0}{60}{anchor=west}{$\theta$} 
            \tdplotdrawarc[tdplot_rotated_coords,-,blue]{(Ch)}{2.074}{119.47}{238.94}{anchor=west}{$109.47^\circ$} 

            \draw (C) -- (G);
    \end{tikzpicture}
    }
    \caption{Schema of the tetrahedral carbon $C_{p_{i}}$ and the possible positions for $C_{p_{i+1}}$, \(H_1\) and \(H_2\). Carbon atoms in greens and hydrogen atoms in Gray.}
    \label{fig:geometry_grid_positive_hydro}
\end{figure}
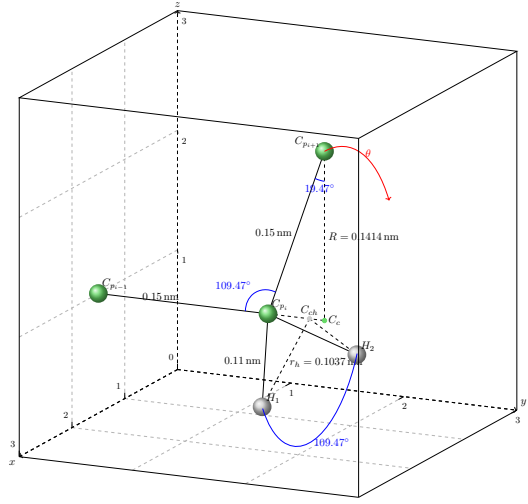

\section{Interconnection tree} \label{appendix:interconnection_tree}

\subparagraph*{Proof of NP-Completeness}

The problem of deciding whether there is an interconnection tree can be generalized to arbitrary  multipartite graphs. 

\begin{problem}[Interconnection Tree Problem (IT)]
Given a multipartite graph $G =( V = V_1\sqcup\dots\cup V_k ,E)$, decide whether there exists an interconnection tree of $G$.
\end{problem}

In contrast to complete multipartite graphs, the problem is $\NP$-complete
for general multipartite graphs.

\begin{theorem} \label{th:ITisNP}
The interconnection tree problem is $\NP$-complete.
\end{theorem}

\begin{proof}
     Let us prove that IT $\in$ NP.
        A witness is any set $T\subseteq E$ of edges of $G$. We can check that $T$ is a matching and that $\partition(T)$ is a spanning tree of $G_{\partition}$ in linear time in $\left| V \right|$. Hence, IT $\in$ NP.

     We now prove that IT is NP-hard by a reduction from the Hamiltonian path problem~\cite{Karp1972}.
        We consider $H=(V_H,E_H)$ an instance of the Hamiltonian path problem.
        
        We build a multipartite graph $G$ such that every vertex of $H$ is represented by a part containing two vertices. More formally, for each $u \in V_H$, we define $V_u  = \{ u^{\text{in}},u^{\text{out}}\}$ and $G =(V,E)$ with
        \begin{itemize}
            \item $V = \left( \cup_{u \in V_H} V_u \right)$
            \item $E = \{ (u^{\text{out}},v^{\text{in}}) \mid (u,v) \in E_H \}$ 
        \end{itemize}
        The graph $G$ is a multipartite graph, where the parts are the $V_{u}$'s. 
        This graph can be built  in time linear in $H$. 

        \textbf{($\Rightarrow$) }
        We prove that if $H$ has a Hamiltonian path then $G$ has an interconnection tree.
        Let $P = u_{1},u_{2},\dots, u_{n}$ be a Hamiltonian path of $H$. 
        We let $T = \{(u_{i}^{\text{out}},u_{i+1}^{\text{in}})\mid 1 \leq t \leq n-1\}$.
         Since $P$ is a Hamiltonian path, no vertex is repeated in the path. 
            Hence, by construction of $T$, its edges have all distinct vertices, which means it is a 
            matching of $G$.
            
            By construction, each edge of $\partition(T)$ is of the form $(V_{u_{i}},V_{u_{i+1}})$.
            Hence, $\partition(T)$ is a path of the quotient graph. Because $P$ is Hamiltonian it contains all vertices of $H$, therefore $\partition(T)$ connect all vertices of the quotient graph $G_{\partition}$. Therefore, $\partition(T)$ is a Hamiltonian path of $G_{\partition}$ and thus is also a spanning tree.

        \textbf{($\Leftarrow$)}
        We prove that if $G$ has an interconnection tree $T$ then $H$ has a Hamiltonian path.
        Let $T$ be a matching of G and $\partition(T)$ is a spanning tree of \(G_{\partition}\).

        Since each part $V_u$ contains exactly two vertices, there are at most two edges of $T$ using a vertex of the same part. Thus, $\partition(T)$ is a connected subgraph of \(G_{\partition}\) and each of its vertices is of degree at most two. Therefore, $\partition(T)$ is a path $V_{u_1},V_{u_2},\dots,V_{u_k}$ connecting all parts of \(G_{\partition}\).
        
         For each consecutive pair $(V_{u_{i}},V_{u_{i+1}})$ in the path $\partition(T)$, the matching $T$ must contain an edge between a vertex of $V_{u_{i}}$ and a vertex of $V_{u_{i+1}}$. By construction of $G$, either all edges between $V_{u_{i}}$ and $V_{u_{i+1}}$ are of the form $(u_{i}^{\text{out}},u_{i+1}^{\text{in}})$ or $(u_{i}^{\text{in}},u_{i+1}^{\text{out}})$. These edges exist in $G$ only if $(u_i,u_{i+1})\in E_H$. Hence, for all $i\leq k$, the pair $(u_i,u_{i+1})$ is an edge in $H$. Therefore, $u_1, u_2,\dots ,u_k$ is a Hamiltonian path in $H$.
\end{proof}

\subparagraph*{Proof of Lemma~\ref{lemma:complete}}
Here, we give a proof of Lemma~\ref{lemma:complete}, which gives a simple characterization of the complete multipartite graphs admitting an interconnection tree. 

\begin{lemma*}[Lemma~\ref{lemma:complete}]
Let \(G = (V = V_1 \sqcup \dots \sqcup V_k, E)\) be a complete multipartite graph.
Then \(G\) admits an interconnection tree if and only if
\(|V| \geq 2(k-1)\).
\end{lemma*}

\begin{proof}
Assume that $T$ is an interconnection tree of $G$. It has $k-1$ edges, because 
$\partition(T)$ is a spanning tree of $G_{\partition}$, which has for vertices the $k$ parts of $G$.
Since $T$ is a matching, no vertex is used twice in $T$, hence there are exactly $2(k-1)$ vertices
in $T$. Thus, $G$ contains at least $2(k-1)$ vertices.

We now prove by induction on $k$ that a complete multipartite graph with $|V| \geq 2(k-1)$ has an interconnection tree. If there is a single part, we always have $|V| \geq 2(1-1) = 0$ then $T = \emptyset$ is an interconnection tree.  If $G$ has exactly two parts $V_1$ and $V_2$, we have 
$v_1 \in V_1$ and $v_2 \in V_2$ because the parts are non-empty and $\{(v_1,v_2)\}$ is an interconnection tree of $G$.

Let us assume the property for $k-1 \geq 2$ parts and let us consider $G$ with $k$ parts such that
$|V| \geq 2(k-1)$. W.l.o.g., let the largest part of $V$ be $V_1$. Because $k\geq 3$ and $|V| \geq 2(k-1)$, we have $|V_1|\geq 2$. We let $u \in V_1$ and 
$v \in V_2$ (which exists because the parts are non-empty).
The multipartite graph $G_{/(u,v)}$ has $k-1$ non-empty parts, and it satisfies $|V(G_{/(u,v)})| = |V(G)| - 2 \geq 2(k-2)$. By the induction hypothesis used on $G_{/(u,v)}$, we have an interconnection tree $T$ of $G_{/(u,v)}$ and by Lemma~\ref{lemma:contraction}, $T \cup \{(v_1,v_2)\}$ is an interconnection tree of $G$, which proves the induction.
\end{proof}

\subparagraph*{Algorithm}

We give here a description of the algorithm described in Section~\ref{sec:interconnection}. The variable $V$ represents the input, it encodes all vertices with their parts stored in order in an array. Function \textsc{Parts(V)} returns the number of parts of $V$ and $|V|$ the number of vertices. 

We use a function \textsc{Update}$(V,u,v)$, which modifies $V$ such that $u$ and $v$
are merged as well as their parts. We also use the function \textsc{Remove}$(V,u)$,
which removes vertex $u$ from $V$. All these operations can be implemented in constant time.

\begin{algorithm}
\scriptsize
\caption{\textsc{Gen\_ICT} \label{alg:ICT}}

\KwIn{
  $V$ \tcp*[f]{Set of vertices with the partition} \\
  $ICT$ \tcp*[f]{Set of solutions, initially empty} \\
  $T$ \tcp*[f]{Current partial interconnection tree, initially empty} \\
  $m$ \tcp*[f]{Margin: $|V| - 2(\textsc{Parts(V)}-1)$}
}

\If{\textsc{Parts(V)}$= 1$}{
  $ICT \gets ICT \cup \{T\}$\;
  \Return
}

\ForEach{$u \in V_1$}{
  \ForEach{$v \in V_i$ \textbf{with} $i > 1$}{
    \textsc{Gen\_ICT}(
      \textsc{Update}$(V,u,v),\,
      ICT,\,
      T \cup \{(u,v)\},\,
      m$
    )\;
  }
  \textsc{Remove}$(V,u)$\;
  $m \gets m - 1$\;
  \If{$m < 0$}{
    \Return
  }
}

\end{algorithm}

Complexity of Algorithm~\ref{alg:ICT} is given in the following theorem,
along with the proof omitted from the main part of the text. 

\begin{theorem*}
Let \(G\) be a complete multipartite graph with \(k\) parts. Then \(\ICT(G)\) can be enumerated with worst-case delay \(O(k)\) and amortized delay \(O(1)\).
\end{theorem*}
\begin{proof}
In Equation~\ref{eq:large_union}, all graphs are complete, hence we can use the algorithm derived from Lemma~\ref{lemma:complete} to test whether they have an interconnection tree. Moreover, the depth is only $k$, since the number of parts is decremented at each recursive call. Hence, we have a delay in $O(nk)$ where $n$ is the number of vertices (we do not need to represent the edges explicitly since the graph is complete multipartite). 

In fact, the delay can be made better by maintaining the proper data structure and doing the union in order of increasing $u$. 
First, the criterion of Lemma~\ref{lemma:complete} can be maintained in constant time. 
We maintain the number of partitions which always decrease by one in a subgraph $(G_{/(u,v)})_{\geq u}$ and the number of vertices which decreases by two because of the contraction and by one each time we consider a contraction from a new vertex.

Second, we can compute $(G_{/(u,v)})_{\geq u}$ from the previous $(G_{/(u,v')})_{\geq u}$ in constant time. To do so, we store the part $V_1$ as a list of the original parts which constitutes it and each original part is stored as a linked list of its element to allow for efficient removal of elements. We also maintain a table of the vertices pointing to the list each vertex belongs to and its position in this list.

When we compute $(G_{/(u,v)})_{\geq u}$ from $(G_{/(u,v')})_{\geq u}$, 
we have to add back the vertex $v'$ to its part and to remove the vertex $v$ in its part. Then, we remove the part of $v'$ to the list representing $V_1$ and insert the part to which $v$ belongs. 
We also need to compute $(G_{/(u,v)})_{\geq u}$ from $(G_{/(u',v')})_{\geq u'}$, where $u$ is the element after $u'$ in the order. The operations are the same as in the previous case, and we also remove $u'$ from its part. The operations are in constant number and in constant time (insertion in a list), hence the complexity for a single recursive call is in $O(1)$. Therefore, the delay of the algorithm is in $O(k)$.

We now analyze the amortized delay of this algorithm. Assume that part $V_1$ is initially chosen to be the largest and will thus always be the largest in the recursive calls. We charge the computation time associated to computing the subgraph $G$ in the enumeration algorithm and to detect the first empty set $\ICT((G_{/(u,v)})_{\geq u})$ to all its children divided equally. This cost is constant because of the previous analysis and denoted by $C$. 

When $|V_1|\geq 2$ then  $|\ICT(G)| \geq 2|\ICT((G_{/(u,v)})_{\geq u})|$, for any $(u,v)$. Moreover, when $|V_1| = 1$, there are one or two parts and the interconnection tree is produced in one step and constant time $C$. Therefore, if we consider an interconnection tree, it has received a cost of $2^{-l}C$ from its ancestor $l$ level upper in the tree. Hence, the cost by interconnection tree is bounded by $\sum_{l\leq k-1}C 2^{-k} \leq 2C$.
\end{proof}